\documentclass[11pt,reqno]{amsart}
\usepackage{amsthm, amsmath,amsfonts,amssymb,euscript,hyperref,graphics,color}
\usepackage{graphicx}
\usepackage{import}
\usepackage{tikz}
\usepackage{latexsym}
\usepackage{cite}%order of citation
\usepackage{verbatim}%comment
\newcommand{\C}[1]{\textcolor{red}{#1}}%specify the area to be changed

\newcommand{\nnb}{\nonumber}

\newcommand{\p}{\partial}

\newcommand{\di}{\mathrm{d}} %{\text{d}}
\newcommand{\dt}{\partial_t}
\newcommand{\tr}{\text{tr}}
\newcommand{\dive}{\text{div}}
\newcommand{\curl}{\text{curl}}
\newcommand{\lie}{\mathcal{L}} 
\newcommand{\dtau}{\partial_{ \tau}}

\newcommand{\D}{\mathcal{D}}

\newcommand{\W}{\mathcal{W}}
\newcommand{\R}{\bar{R}}
\newcommand{\Ew}{\mathcal{E}}
\newcommand{\Hw}{\mathcal{H}}
\newcommand{\tE}{\tilde{E}}
\newcommand{\Tr}{ {}^R{\breve T} }

\newcommand{\be}{\begin{equation}}
\newcommand{\ee}{\end{equation}}

\newcommand{\al}[2]{
	\begin{align}\label{E:#1}
		#2
	\end{align}
}
\newcommand{\als}[1]{
	\begin{align*}
		#1
	\end{align*}
}

\newtheorem{theorem}{Theorem}[section]
\newtheorem{lemma}[theorem]{Lemma}
\newtheorem{proposition}[theorem]{Proposition}
\newtheorem{corollary}[theorem]{Corollary}

\newtheorem{remark}[theorem]{Remark}

\numberwithin{equation}{section}

\allowdisplaybreaks[4]

\begin{document}
\title[Global stability of the $1+3$ open Milne]{Global stability of the open Milne spacetime}

\author[J. Wang]{Jinhua Wang} 
\email{wangjinhua@xmu.edu.cn}
\address{School of Mathematical Sciences, Xiamen University, Xiamen 361005, China}

\author[W. Yuan]{Wei Yuan} 
\email{yuanw9@mail.sysu.edu.cn}
\address{Department of Mathematics, Sun Yat-sen University, Guangzhou, China}
%\keywords{Open Milne model, Einstein scalar field equations, Nonlinear stability; MSC2010: 35Q76, 83C05.}

\begin{abstract}
The open Milne cosmological spacetime has a 3-dimensional Cauchy surface isometric to the (non-compact) hyperbolic space. We prove the globally nonlinear stability of the open Milne spacetime for both  massive and massless Einstein-scalar field equations and show that as time goes to infinity, the spatial metric tends to the hyperbolic metric. 
%We are able to establish the energy decay estimates in geodesic polar gauge with lower regularity for the trace part of the second fundamental form $\tr k$ than that in the local theory \cite{Fournodavlos-Luk-Kasner}. This allows us to retrieve the top order of $\tr k$ with uniform bound (without decay) when the scalar field is massive. Meanwhile,  we derive almost $t^{-1}$ decay estimates for the geometry and massless scalar field, using the spectrum theory for Laplacian on ${\bf H^3}$. 
%The approach we proposed is independent of the CMC foliation and lower bound for the eigenvalues of Einstein operator. Instead, 
The proof is based on the Gaussian normal coordinates, in which the decay rates of gravity are determined by the expanding geometry of Milne spacetime. 
\end{abstract}
\maketitle
%\tableofcontents

\section{Introduction}\label{sec-intro}

%The backgrounds that we consider are the following family of cosmological vacuum spacetimes. 
Let $M^3$ be a complete Riemannian manifold differmorphic to $\mathbb{R}^3$ and suppose it admits an Einstein metric $\gamma$ with negative Einstein constant. In 3-dimensional case, $\gamma$ is a hyperbolic metric with constant sectional curvature.  Without losing of generality, we can assume $\gamma$ has sectional curvature $-1$ simply by rescaling the metric.
In particular, we consider $(M^3, \gamma)$ is given by a hyperboloid ${\bf H^3}$ with the induced metric in Minkowski spacetime. The hyperbolic space is unique up to isometry due to the classification theorem of space forms. %(since we know that there is only one simply-connected Riemannian manifold up to isometry which has constant sectional curvature $-1$). 

Let $\mathcal{M}^4$ be a $4$-manifold of the form $\mathbb{R} \times M^3$. %where $I \subset \mathbb{R}$ is an interval,   
Then $(\mathcal{M}^4, \bar\gamma)$ with $\bar\gamma$ $$\bar\gamma = -dt^2 + t^2 \gamma$$ is a solution to the vacuum Einstein equations and known as the $(1+3)$-dimensional \emph{Milne model}. The Milne model is also known as the $\kappa=-1$ vacuum Friedmann-Robertson-Walker (FRW) model. We  call it the \emph{open Milne model} to distinguish it from the \emph{closed Milne spacetime} mentioned below.

By passing to a quotient of the open Milne spacetime, one obtains the closed Milne model, which is a cosmological vacuum spacetime admitting a Cauchy surface isometric to a 3-dimensional closed (compact without boundary) hyperbolic manifold. Andersson and Moncrief \cite{A-M-04} first proved the nonlinear stability of closed Milne model based on the constant mean curvature, spatially harmonic (CMCSH) gauge \cite{A-M-03-local}. Later, they \cite{A-M-11-cmc} generalized the stability result to a family of $(1+n)$--dimensional, spatially compact spacetimes whose spatial manifolds are stable, compact, negative Einstein spaces. The global stability of closed Milne spacetime was further investigated in variously non-vacuum context \cite{Andersson-Fajman-20, Fajman-Wyatt-EKG, Wang-J-EKG-2019, Barzegar-Fajman-charged, Faj-Urban-CC}. In addition, there were results considering the Kaluza-Klein spacetime built on the closed Milne model \cite{Branding-Fajman-Kroencke-18, Wang-KK-21}.

Related to the open Milne spacetime is the Minkowski spacetime, a flat solution of the vacuum Einstein equations. When restricted on the future of a hyperboloid in Minkowski spacetime, % Note that a hyperboloid is conformally compact. In general, it is technically useful to work with the conformal compactification of an asymptotically hyperbolic manifold. 
Friedrich \cite{Friedrich-91} used a conformal method to show that asymptotically simple hyperboloidal initial data close to Minkowskian hyperboloidal initial data evolve into a global solution of the $(1+3)$-dimensional Einstein-Maxwell-Yang-Mills equations. Furthermore, this solution has a similar asymptotic behavior as the Minkowski space.  The existence of such smooth hyperboloidal initial data was shown in \cite{hyper-data-LCF}. As a remark, the conformal method proof  \cite{Friedrich-91} relies on the conformal invariance of the $(1+3)$-dimensional Maxwell-Yang-Mills equations (i.e. the trace-free feature of energy momentum tensor in a 4-dimensional spacetime). 
The global stability of Minkowski spacetime with asymptotically flat data was proved by Christodoulou and Klainerman \cite{Christodoulou-K-93} through a covariant proof based on Bel-Robinson energy. Later, alternative proofs using wave coordinates were given by Lindblad and Rodnianski \cite{Lind-Rod-05, Lind-Rod-10}. Along with these works, more proofs for the Einstein equations with matter fields had been developed recently, see for instance \cite{Ionescu-EKG, Ma-Lefloch-EKG, Taylor-17, Lind-Taylor-17, F-J-S-17, Fajman-vlasov-2}.

In this paper, we concern the globally nonlinear stability of the open Milne spacetime %(or in other words, stability of the 3-dimensional non-compact hyperbolic space) 
for the Einstein-scalar field equations. %Note that, the scalar field equation is non-conformal invariant. %and hence it would be difficult to apply the conformal method of \cite{Friedrich-91}.
Stability problem for general $n$-dimensional non-compact, negative Einstein space will be considered in another article.

\subsection{Main result}
The $(1+3)$-dimensional Einstein-scalar field equations on the spacetime manifold $(\mathcal M, \, \breve g)$ take the form of
\begin{subequations}
\begin{equation}
\breve R_{\alpha \beta} - \frac{1}{2} \breve R \breve g_{\alpha \beta} = \breve{T}_{\alpha \beta}(\phi),   \label{eq-Einstein-source}
\end{equation}
\begin{equation}
 \breve{T}_{\alpha \beta}(\phi)= \breve D_\alpha \phi \breve D_\beta \phi - \frac{1}{2} \breve g_{\alpha \beta} \left( \breve D^\mu \phi  \breve D_\mu \phi + m^2 \phi^2 \right), \label{def-energy-Mom-kg}
 \end{equation}
 \end{subequations}
where $\breve R_{\alpha \beta}$ and $\breve R$ denote the Ricci and scalar curvature of the spacetime
metric $\breve g_{\alpha \beta}$ respectively, $\breve D$ is the covariant derivative associated to $\breve g_{\alpha \beta}$ and $\breve{T}_{\alpha \beta}(\phi)$ is the energy momentum tensor for the scalar field $\phi$. 
The Bianchi identities 
imply that the scalar field $\phi$ satisfies 
\begin{equation}\label{eq-kg}
\Box_{\breve g} \phi - m^2 \phi = 0,
\end{equation}
with the Laplacian operator given by $\Box_{\breve g} = \breve D^\alpha \breve D_\alpha$.  For $m^2>0$, $\phi$ satisfying \eqref{eq-kg} is called the massive scalar field (or Klein--Gordon field), otherwise $m^2=0$, it is referred as a massless scalar field. We remark that the energy momentum tensor \eqref{def-energy-Mom-kg} is not trace-free and hence the scalar field equation is not conformal invariant.

Before the statement of our main results, we introduce some notations. 
Throughout the paper, Greek indices $\alpha, \beta \cdots, \mu, \nu \cdots$ run over $0,\cdots,3$.  Latin indices $i,j, \cdots$ run over $1,\cdots, 3$. On the spacetime manifold $(\mathcal{M}^4 = \mathbb{R} \times M^3, \, \breve g)$, the spacetime metric $\breve{g}$ in geodesic polar coordinates (or Gaussian normal coordinates) takes the form of
\begin{equation}\label{metric-form-tau}
\breve g_{\mu\nu} = - \di t^2 + \tilde g_{ij} \di x^i \di x^j.
\end{equation}
The spatial metric $\tilde g_{ij}$ is the induced metric of $\breve g_{\mu \nu}$ on the spatial manifold $M_t := \{t\} \times M^3$. Let \[ \tilde k_{ij}=-\frac{1}{2} \mathcal{L}_{\p_t} \tilde g_{ij} \] be the corresponding second fundamental form. We define the normalized variables
\begin{equation}\label{rescale-metric}
g_{ij} =t^{-2} \tilde g_{ij},  \quad k_{ij}  =t^{-1} \tilde k_{ij},
\end{equation}
and denote the traceless and (perturbed) trace parts of $k_{ij}$ to be
\be\label{rescale-2}
 \Sigma_{ij} = k_{ij} - \frac{\tr_g k }{3} g_{ij}, \quad \eta = \frac{\tr_g k}{3} + 1,
\ee
where $\tr_g k=g^{ij} k_{ij} = t  \tilde g^{ij} \tilde k_{ij}$, and $\tr_g$ refers to the operation of taking trace with respect to the metric $g$. Throughout this paper, we denote $\nabla$ to be the connection corresponding to $g$, and $R_{ij}$ be its Ricci curvature.
The notation $x \lesssim y$ means $x \leq Cy$ for some universal constant $C$, and $x \sim y$ means both $x \lesssim y$ and $y \lesssim x$ hold. We also use the notation $\lesssim_N$ to point out the dependence of the constant $C$ on some fixed constant $N$. The Sobolev space with respect to $g$ on $M$, $H_{k} (M, g)$, is abbreviated as $H_k$, and $\| \cdot \|$ means the norm of $L^2(M,g)$.

Let $(M^3, \, g_0)$ be a smooth complete Riemannian manifold diffeomorphic to $\mathbb{R}^3$ with positive injectivity radius. 
\begin{theorem}\label{thm-global-existence}
Assume that $(M^3,  \, g_0, \,  k_0, \, \phi_0, \, \phi_1)$ is a rescaled data set for the Einstein scalar field equations % \eqref{eq-Einstein-source}-\eqref{def-energy-Mom-kg}, \eqref{eq-kg}, 
where $\phi_0 = \phi|_{t=t_0}, \, \phi_1 = t \p_t \phi|_{t =t_0}$. There is a constant $\varepsilon > 0,$ so that if for some fixed integer $N \geq 2$,
\begin{align} 
& \|g_0-\gamma\|_{H_{N+1}(M, \,g_0)} + \|  k_{0} + g_{0} \|^2_{H_{N+1}(M, \,g_0)} +   \| \tr_{g_0} k_{0} + 3 \|^2_{H_{N+2}(M, \, g_0)}   \nnb \\
&  +  \| R_{0ij} + 2 g_{0ij} \|^2_{H_N(M, \, g_0)} +  \|\phi_0\|^2_{H_{N+2}(M, \, g_0)} +  \|\phi_1\|^2_{H_{N+1}(M,  \,g_0)}  \leq \varepsilon^2, \label{intro-initial-data}
\end{align}
then $(\mathcal{M}^4, \, \breve g)$ with $\breve g = -dt^2 + t^2 g$, $t \geq t_0$ is a global solution of the Einstein scalar field system, where $ R_{0ij}$ denotes the Ricci curvature of $ g_0$. Moreover, the spatial metric $g(t)$ tends to the hyperbolic metric $\gamma$, as $t$ goes to infinity.

In the case of massive Einstein scalar field $(m^2>0)$, we have 
\als
{
& t^{1-\delta} ( \|\Sigma_{ij} \|_{H_{N+1}} + \|R_{ij} + 2g_{ij}\|_{H_N}  ) +  t^{1-C_N \varepsilon} \|\eta \|_{H_{N+1}}+ \|\eta \|_{H_{N+2}}  \nnb \\
&\qquad + \|g_{ij} - \gamma_{ij} \|_{H_{N+1}} + t^{\frac{1}{2}} ( \|\phi\|_{H_{N+2} } +  \|t \p_t \phi\|_{H_{N+1} } ) \lesssim_N \varepsilon.
}
In the case of massless Einstein scalar field $(m^2=0)$, 
\als
{
& t^{1-\delta} (  \|\Sigma_{ij} \|_{H_{N+1}} + \|R_{ij} + 2g_{ij}\|_{H_N}  ) + t \|\eta \|_{H_{N+2}} \nnb \\
&\qquad + \|g_{ij} - \gamma_{ij} \|_{H_{N+1}} + t^{1-\sqrt{\sigma}} ( \|\phi\|_{H_{N+2} } +  \|t \p_t \phi\|_{H_{N+1} } ) \lesssim_N \varepsilon.
}
Here $0<\delta, \, \sigma<1/6$ are two fixed, independent constants, and the constant $C_N$ depends on $N$ (in particular, not on $\varepsilon$ and $t_0$).

\end{theorem}
\begin{remark}
As suggested in the proof, Theorem \ref{thm-global-existence} holds automatically if $M$ is a compact manifold without boundary. 
\end{remark}

\begin{remark}
As shown in the proof, it holds that $\| t\p_t g_{ij} \|_{H_{N+1}}  \lesssim_N \varepsilon  t^{\delta-1}$.  Together with $\|R_{ij} + 2g_{ij}\|_{H_N } \lesssim \varepsilon t^{-1+\delta}$, it implies $g(t)$ tends to an Einstein metric (which is the hyperbolic metric $\gamma$ in dimension three) as $t \rightarrow +\infty$. In fact, we will see it in the proof that $\| R_{imjn} + g_{ij } g_{m n} - g_{in} g_{jm} \|_{H_{N} } \leq \varepsilon t^{-1+\delta}$ holds. 
Note that on a 3-dimensional Riemannian manifold, the Ricci curvature determines the full curvature tensor and hence it is Einstein if and only if it has constant sectional curvature \cite{Besse-Einstein}.
More elaborate rigidity results for the class of asymptotically hyperbolic manifolds could be found for instance in \cite{Shi-Tian-AH,H-Qing-Shi-AH}.
\end{remark}

\begin{remark}
The Milne geometry exhibits $t^3$ volume growth, which suppresses the formation of singularity. As a consequence, it suffices to carry out the energy estimates for the long time dynamics in the standard (unweighted) Sobolev space.  Note that, solutions of the constraint equations were constructed in some weighted Sobolev spaces \cite{hyper-data-LCF,Andersson-Ch-AH,Gi-Sa-AH,Isen-Park-AH} for asymptotically hyperbolic manifolds. Therefore, it would be interesting to work out the energy estimates in some weighted Sobolev space to give more control on the behavior of the solution at infinity.
\end{remark}

\subsection{Related works}
Let us begin with reviewing some relevant works and approaches on the closed Milne model. In the proof for general $(1+n)$-dimensional case with the spatial manifold $(M^3, \gamma)$ being a stable, negative Einstein space of compact type, Andersson-Moncrief \cite{A-M-11-cmc} introduced the CMCSH gauge and developed energy estimates through a wave type energy  for the perturbed metric. Moreover, it was shown in \cite{A-M-11-cmc} that the decay rates of gravity depend on the stability properties of the Einstein geometry. Combined with a sharp estimate on the lower bound of eigenvalues $\lambda$ of the Einstein operator \cite{Kroncke-15} which yields $\lambda \geq 1$ when restricted in the $3$-dimensional negative Einstein space, the method of \cite{A-M-11-cmc} provided $t^{-1+\varepsilon}$ decay estimates in $(1+3)$-dimensional case and had prompted more works in non-vacuum context \cite{Andersson-Fajman-20, Branding-Fajman-Kroencke-18,Fajman-Wyatt-EKG, Barzegar-Fajman-charged, F-O-O-W-fluid23}. 
\begin{comment}
We remark that in views of the literature \cite{A-M-11-cmc, Gerhardt-CMC-open} and \cite{LY-2010}, where a Poincar\'{e} type inequality was derived on the non-compact hyperbolic manifold, it is possible to give a proof for the stability of open Milne spacetime (or stability of hyperbolic metric $\gamma$ for the Einstein flow) using the CMCSH gauge. However, this might be difficult because in the noncompact case, we are not restricted to the TT--tensor, which is important in the proof of Andersson-Moncrief.
\end{comment}

Alternatively, when considering the Einstein Klein--Gordon system, the proof of \cite{Wang-J-EKG-2019} is based on the CMC with zero shift gauge and Bel-Robinson energy (cf. \cite{A-M-04,Christodoulou-K-93}). In this framework \cite{Wang-J-EKG-2019}, energy estimates are built on the Bianchi equations complemented with an elliptic equation for the lapse, and an elliptic, div-curl system together with a transport equation for the second fundamental form.  In particular, the $t^{-1+\varepsilon }$ decay estimate for the gravity in \cite{Wang-J-EKG-2019} was due to the expanding geometry of Milne spacetime which had been reflected on the structure of rescaled Bianchi equations. More specifically, there are linear terms with favourable signs on the left hand side of the rescaled Bianchi equations \eqref{eq-1+3-bianchi-T-E}--\eqref{eq-1+3-bianchi-T-H}, and the coefficients of these linear terms determine the decay rates. Similarly, the expanding geometry is reflected on the transport equation for the second fundamental form $\Sigma$ as well, see \eqref{eq-evolution-2}. 
We know that, with the CMC gauge,  the main difficulty in \cite{Wang-J-EKG-2019} comes from the borderline terms arising from the massive term in the Klein--Gordon equation which, if taken higher derivatives, is coupled with commutators between the spatial derivatives $\nabla$ and lapse. As a consequence, \cite{Wang-J-EKG-2019} had made significant efforts to linearize these nonlinear borderline terms, including refined estimates for the lapse, $L^\infty-L^\infty$ estimate and hierarchies of energy estimates for the Klein--Gordon field. Motivated by these facts, it is natural to work in a gauge with constant lapse (for instance, the Gaussian time coordinate where the lapse is fixed to be constantly one), so that there would be no such borderline terms and the proof of \cite{Wang-J-EKG-2019} can be significantly simplified. 

On the other hand, continued with the CMC foliations, the existence of CMC foliations in spacetimes with non-compact spatial manifold is not known until recently Gerhardt \cite{Gerhardt-CMC-open} gave the proof in a spacetime asymptotic to the open FRW spacetime with spatial curvature $\kappa=0, \, -1$.
However, we prefer the Gaussian gauge for the above reason of simplifications.
Gaussian normal coordinate is known as the geodesic polar coordinate \eqref{metric-form-tau} or synchronous coordinate.  Local existence theorem for the vacuum Einstein equations in geodesic polar gauge can be implied by the work of Andersson-Rendall \cite{Andersson-Rendall-quiescent} in the analytic category, while for data with bounded energy, it was recently addressed by Fournodavlos-Luk \cite{Fournodavlos:2021aa, Fournodavlos-Luk-Kasner}. In particular, Fournodavlos-Luk  \cite{Fournodavlos-Luk-Kasner} successfully established an energy argument through the equivalent, reduced system  \eqref{E:reduced-sys}: a second order equation for the second fundamental form $\tilde k$ coupled with transport equations for $h:=\tr_{\tilde g} \tilde k$ (viewed as an independent variable) and the metric $\tilde g$. Since the second order equation for $\tilde k$ contains the second derivative of $h$, this reduced system has an apparent issue of regularity. Remarkably, it was resolved by means of renormalizations together with elliptic estimates \cite{Fournodavlos-Luk-Kasner} (referring to Section \ref{sec-improve-regularity} as well). Hence the local existence theorem in \cite{Fournodavlos-Luk-Kasner} was achieved in the space $(\tilde g, h, \tilde k, \p_t \tilde k) \in H_{N+2} \times H_{N+2} \times H_{N+1} \times H_{N}$, $N\geq 2$.  

\subsection{Comments on the proof}
Based on \cite{Wang-J-EKG-2019, Fournodavlos-Luk-Kasner}, we will establish an approach that is independent of the CMC foliation and the theory of infinitesimal Einstein deformations, so that the stability theorem holds regardless of whether the spatial manifold is compact or not. In addition, the analysis is independent of lower bound for the eigenvalues of Einstein operator or certain Poincar\'{e} type inequality. While it only depends on the expanding geometry of the spacetime and thus is promising to be extended to higher dimensional case. 

In this article, we mainly adopt the Gaussian normal coordinates (or geodesic polar coordinates). It is known that such a gauge leads to the equation $\p_t \tr \tilde k = |\tilde k|^2 \geq \frac{1}{3} (\tr \tilde k)^2$. In contrast to the asymptotically flat case \cite{Christodoulou-K-93}, nevertheless, $\tr \tilde k$ will not blow up in finite time since it has negative sign on the initial hypersurface close to the hyperbolic manifold. This implies the geodesic polar gauge will be non-singular during the evolution. Moreover, it allows us to take non-CMC data in the non-compact case straightforwardly\footnote{In the compact case, one removes the CMC restriction on data through the strategy of taking advantage of the existence of local development for data close to CMC \cite{Ringstrom-09} and finding a CMC surface in such a local development \cite{Fajman-Kroencke-15}. For asymptotically flat case, one can refer \cite{Bartnik-Maximal, Christodoulou-K-93} for similar ideas. %The existence of CMC surface is stated in \cite{Andersson-Iri} for a class of asymptotically Schwarzschild spacetime.   
}.

In the approach of geodesic polar gauge, it is difficult for us to work out the long time energy estimates based on the second order equation for $\tilde k$ (which equivalently gives a second order equation for $\Sigma$ \eqref{E:wave-Sigma}). %Compared to \cite{Fournodavlos-Luk-Kasner} where the spatial manifold of the background spacetime is a flat torus, we have negative spatial curvature, which gives rise to a linear term $2\Sigma$ with an unfavourable sign on the right hand side of \eqref{E:wave-Sigma}. Since we prefer our analysis to be general and thus independent of some Poincar\'{e} type inequality, this unfavourable linear term requires us to derive decay estimate for $\|\Sigma\|_{L^2}$ ahead. However, in this gauge, the transport equation for $\Sigma$ has shown its dependence on the spacetime curvature, which means that we are unable to avoid estimating the spacetime curvature. Therefore it will be more straightforward to work with the Bianchi equations to estimate the Bel-Robinson energy for the curvature.
Since when the massive scalar field is present, we can only obtain a uniform bound (without decay) for the top order derivative term $\nabla^{N+2} \eta$ (Section \ref{sec-improve-regularity}), which will eventually leads to a logarithmic growth for the top order derivative metric $\nabla^{N+2} (g-\gamma)$. In view of the wave equation of $\Sigma$ \eqref{E:wave-Sigma}, the non-decaying $\nabla^{N+2} \eta$ and growing $\nabla^{N+2} (g-\gamma)$ will be the main obstructions to close the energy argument for this reduced system \eqref{E:reduced-sys} in the long time dynamics. 

In view of the above observations, we follow the approach of Bel-Robinson energy, working on the Bianchi equations coupled with equations for $\eta$ and $\Sigma$. %which is relatively easier to be adapted to the case of non-compact spatial manifolds. 
Compared to the reduced equations \eqref{E:reduced-sys}, this original system of Einstein equations have the advantages of ``decoupling'' lower order derivatives of $\eta$ from the top order one. In other words, due to a structure of saving regularity for the transport equation of $\eta$, we are able to carry out the main energy estimates for the original system with the regularity $(\eta,\, \Sigma) \in H_{N+1} \times H_{N+1}$ (which will be explained in the next paragraph). On the contrary, in the local existence theory \cite{Fournodavlos-Luk-Kasner}, one needs $(\eta, \,\Sigma) \in H_{N+2} \times H_{N+1}$ instead when focusing on the reduced system. Once the main energy argument is established, we can improve the regularity of $\eta$ by retrieving the estimate for $\nabla^{N+2}\eta$ in the same manner as \cite{Fournodavlos-Luk-Kasner} (referring to Section \ref{sec-improve-regularity} as well). With this strategy, although the term $\nabla^{N+2}\eta$ is merely uniformly bounded without decay in the massive case, it will not interrupt the main energy estimates and we can still obtain proper decay estimates for other variables. 

Meanwhile, unlike the CMC gauge \cite{Wang-J-EKG-2019}, the geodesic polar gauge has an obvious difficulty: the term $\nabla \tr_g k = \nabla \eta$ does not vanish, and thus the constraint equations for the second fundamental form $k$ \eqref{E:div-curl-k-hat} fail to be elliptic. The resolution of this difficulty lies in an observation that the transport equation \eqref{eq-evolution-xi} of $\eta$,
\be\label{eq-eta-into}
t \p_{t} \eta + \eta = \eta^2 + \frac{1}{3} |\Sigma|^2 + \frac{1}{3} t^2 \breve R_{t t},
\ee
admits a structure of saving regularity. That is, as the second derivative of metric, $\p_t \eta$  is identical to terms on the first derivative of metric, such as $\eta^2$, $\Sigma^2$, and $\breve R_{t t} = (\p_t \phi)^2 - m^2 \phi^2$ due to the Einstein-scalar field equations \eqref{eq-Einstein-source}-\eqref{def-energy-Mom-kg}. Note that,  $\p_t \phi$ has the same regularity as the first derivative of metric when the field equations being viewed roughly as a wave system for the metric and $\phi$.
Thus \eqref{eq-eta-into} helps us deriving a decay estimate for $\|\eta\|_{H_{N+1}}$ a priori. After that, with $\nabla \eta$ being viewed as a bounded source term, the div-curl system becomes \eqref{E:div-curl-k-hat} is an elliptic system for the traceless part $\Sigma$. These elliptic estimates together with the decay estimate for $\|\Sigma\|_{L^2}$, which follows from the transport equation of $\Sigma$, give rise to the decay estimate for $\|\Sigma\|_{H_{N+1}}$. Roughly speaking, the most difficult part in \cite{Wang-J-EKG-2019}, hierarchies of estimates between the lapse and Klein--Gordon field, are replaced by hierarchies between the estimates for $\eta$ and $\Sigma$, which technically simplifies the proof of \cite{Wang-J-EKG-2019} as well. Moreover, since the $L^\infty-L^\infty$ estimate for the Klein--Gordon field is no longer needed now, we are allowed to complete the proof with one order lower regularity with $N \geq 2$ than $N = 3$ in \cite{Wang-J-EKG-2019}.

%Along with the long time energy estimates, the local existence theorem has to be reformulated, since we intend to get rid of the top order $\nabla^{N+2} (g-\gamma)$ which will have energy growth in the massive case. %In fact, the $H_{N+2}$ norm of the metric can be replaced by $H_N$ norm of the Ricci curvature \cite{Christodoulou-K-93}. 
%Motivated by \cite{Christodoulou-K-93} and \cite{Fournodavlos-Luk-Kasner}, we have modified local existence theorem in Appendix \ref{sec-local}  for practice.

Next, we turn to the massless case. When the geometry of  spacetime is (or close to) the closed Milne spacetime, the massless scalar field rarely decays \cite{Branding-Fajman-Kroencke-18, Wang-KK-21}, since the lowest eigenvalue of Laplacian on a closed hyperbolic manifold can be arbitrarily small. Note that, with a spectral assumption $\lambda (-\Delta_{\gamma}) >1$ for Laplacian operator on the compact spatial manifold, Fajman and Urban \cite[Section 9]{Faj-Urban-CC} had obtained some decay rates for the scalar field. Fortunately, for non-compact hyperbolic manifold ${\bf H^3}$, the eigenvalue of Laplacian operator has a positive lower bound ($\lambda (-\Delta_{\gamma}) \geq 1$) \cite[Chapter II, Theorem 5]{Chavel}. Based on this fact, and followed the idea of modified, wave type energy \cite{A-M-11-cmc}, we can derive almost $t^{-1}$ decay for both the gravity and massless scalar field. 
As a remark, unlike the massive case, one is able to prove the $t^{-1}$ decay estimate for $\|\eta\|_{H_{N+2}}$  in the massless case. 

Compared with the closed Milne model, where the spatial manifold is compact without boundary, one of the differences in the non-compact case lies in that we need to establish density theorems based on proper regularity of the metric (see Proposition \ref{pro-density} and its applications, corollaries \ref{coro-density}--\ref{prop-elliptic-Delta-1}). In fact, we only have uniform bound for lower regularity ($N+1$ derivatives) of the metric instead of $N+2$ derivatives as usual.

The article is organized as follows. In Section  \ref{preliminaries}, we introduce some relevant notations, and geometric equations. In the sections \ref{sec-EKG} and \ref{sec-EW}, we establish the energy argument for the massive and massless Einstein-scalar field equations respectively. In the end, we collect the local existence theorem, the density theorem and some geometric identities in the appendix.  

{\bf Acknowledgement} J.W. is supported by NSFC (Grant No. 12271450 and No. 11701482). W.Y. is supported by NSFC (Grant No. 12071489 and No. 12025109).

\section{Preliminary}\label{preliminaries}

\subsection{Lorenzian geometric equations}\label{sec-einstein-eq}

Recall the spacetime metric $\breve{g}_{\mu \nu}$ in geodesic polar coordinate \eqref{metric-form-tau} and the rescaled variables \eqref{rescale-metric}--\eqref{rescale-2} on the spatial manifold $M_t:= \{t\} \times M^3$. In particular, $g_{ij}=t^{-2} \tilde g_{ij}$ is the rescaled spatial metric, and $\nabla$ is the corresponding connection. It implies
\[g^{ij}=t^{2} \tilde g^{ij}, \quad   \di \mu_g =t^{-3} \di \mu_{\tilde g},\]
and the rescaled curvatures are given by
\begin{align*}
 R_{imjn} &= t^{-2} \tilde R_{imjn}, \quad  R_{ij } = \tilde R_{ij}, \quad   R = t^2 \tilde R.
\end{align*}
The notations $R_{imjn}$ and $ \tilde R_{imjn}$ denote the Rieman tensors with respect to $g$ and $\tilde g$ respectively.
We  introduce the logarithmic time 
\begin{equation}\label{def-tau}
\tau := \ln t,
\end{equation}
so that $\p_{\tau} = t\p_t$.
For notational convenience, we introduce the rescaled spacetime metric \[\bar g_{\mu \nu} =t^{-2} \breve{g}_{\mu \nu} = - \di \tau^2 + g_{ij} \di x^i \di x^j,  \quad \bar g^{\mu \nu}=t^{2} \breve{g}^{\mu \nu}, \]
 and define
\begin{align*}
%\bar g_{\mu \nu} &=t^{-2} \breve g_{\mu \nu}, \quad \quad \,\,\, \bar g^{\mu \nu}=t^{2} \breve g^{\mu \nu}, \\
\bar R_{\mu \alpha \nu \beta} &= t^{-2}  \breve R_{\mu \alpha \nu \beta}, \quad \bar R_{\mu  \nu } = \breve R_{\mu  \nu}, %\quad  \,\, \bar R = t^2 \breve R.
\end{align*}
where $\breve R_{\mu \alpha \nu \beta}$,  $\breve R_{\mu  \nu}$ denote the Rieman and Ricci tensors with respect to $\breve g$.

The geometric structure equations are given by
\begin{subequations}
\begin{align}
\mathcal{L}_{\dtau} g_{ij} &=-2 \eta g_{ij} - 2 \Sigma_{ij},  \label{eq-evolution-1} \\
\p_{\tau} \eta + \eta &= \eta^2 + \frac{1}{3} |\Sigma|^2 + \frac{1}{3} \R_{\tau \tau}, \label{eq-evolution-xi} \\
\mathcal{L}_{\dtau} \Sigma_{ij} + \Sigma_{ij} &=  \R_{i \tau j \tau} - \Sigma_{ip} \Sigma_{j}^p -  \frac{1}{3} (\R_{\tau \tau} + |\Sigma|^2) g_{ij}. \label{eq-evolution-2}
\end{align}
\end{subequations}
We also have the Gauss equations
\al{Gauss-Riem-hat-k}
{
R_{imjn} &= -  \frac{1}{2} ( g \odot g )_{imjn} + ( \eta -  \frac{\eta^2}{2} ) ( g \odot g )_{imjn}  \nnb \\
& + ( 1- \eta ) ( g \odot \Sigma )_{imjn} - \frac{1}{2} ( \Sigma \odot \Sigma )_{imjn}   + \R_{imjn},
}
and the Codazzi equations
\begin{equation}
\nabla_i \Sigma_{jm} -\nabla_j \Sigma_{im} + \nabla_i \eta g_{jm} - \nabla_j \eta g_{i m}=\R_{\tau mij}, \label{Codazzi-curl-k}
\end{equation}
where $\odot$ is  the \emph{Kulkarni-Nomizu product} 
\begin{equation*}%\label{def-odot}
( \xi \odot \zeta )_{imjn} = \xi_{ij} \zeta_{mn} - \xi_{in} \zeta_{jm} + \zeta_{ij} \xi_{mn} - \zeta_{in} \xi_{jm},
\end{equation*}
for any symmetric $(0, 2)$-tensors $\xi$ and $\zeta$.
Taking contractions on \eqref{E:Gauss-Riem-hat-k} and \eqref{Codazzi-curl-k} leads to
\begin{subequations}
\begin{align}
R_{ij} +  2 g_{ij} &= ( 4 \eta - 2 \eta^2 ) g_{ij} + \Sigma_{ip} \Sigma_{j}^p + ( 1 - \eta ) \Sigma_{ij} + \R_{i \tau j \tau} + \R_{ij},\label{Gauss-Ricci-hat-k}\\
R +  6 &= 12 \eta - 6 \eta^2  + \Sigma_{ij} \Sigma^{ij}  + 2 \R_{\tau \tau} + \R, \label{Gauss-Ricci-trace-hat -k} \\
\nabla^i \Sigma_{ij} - 2 \nabla_j \eta &= - \R_{\tau j}. \label{Codazzi-div-k} 
\end{align}
\end{subequations}
In view of \eqref{Gauss-Ricci-hat-k}, \eqref{eq-evolution-2} can be rewritten  alternatively as
\al{eq-trans-Sigma}{
\dtau \Sigma_{ij} +2 \Sigma_{ij} &= R_{ij} +  2 g_{ij}  + \eta \Sigma_{ij}
- 2 \Sigma_{ip} \Sigma_{j}^p \nnb \\
& - \R_{ij} -  \frac{1}{3} (\R_{\tau \tau} + |\Sigma|^2) g_{ij} - ( 4 \eta - 2 \eta^2 ) g_{ij}.
} 

We next discuss the Bianchi equations.
Let $\mathcal{W}_{\alpha \beta \gamma \delta}$ be the Weyl tensor of $\breve R_{\alpha \beta \gamma \delta}$, and
\begin{equation}\label{eq-Bianchi-J}
{}^{*}\mathcal{W}_{\alpha \beta \gamma \delta} = \frac{1}{2} \epsilon_{\alpha \beta \mu \nu} \mathcal{W}^{\mu \nu}{}_{\! \gamma \delta}.
%\quad  \mathcal{W}^*_{\alpha \beta \gamma \delta} = \frac{1}{2} \mathcal{W}_{\alpha \beta}{}^{\! \mu \nu} \epsilon_{\mu \nu \gamma \delta}.
\end{equation}
be the Hodge dual of $\mathcal{W}$.
The electric and magnetic parts $\Ew, \, \Hw$ of the Weyl field $\mathcal{W}$, with respect to the foliation $\{M_t\}$ are defined by
\begin{equation}\label{def-electric-magnetic}
\Ew_{\alpha \beta} = \mathcal{W}_{\alpha \mu \beta \nu} \p_{t}^\mu \otimes \p_{t}^\nu, \quad \Hw_{\alpha \beta} = {}^{*}\mathcal{W}_{\alpha \mu \beta \nu} \p_{t}^\mu \otimes \p_{t}^\nu.
\end{equation} 
We shall allow ourselves to use $\W$ to represent the electric part $\Ew$ or the magnetic part $\Hw$.
The $1+3$ rescaled Bianchi equations then read as follows \cite[Proposition 7.2.1]{Christodoulou-K-93}, \cite[Corollary 3.2]{A-M-04},  
\begin{subequations}
\begin{equation}
\begin{split}
 & \lie_{\dtau} \Ew_{ij} - \curl \Hw_{ij} + \Ew_{ij} \\
   = {} & \eta \Ew_{ij}  - \frac{5}{2} \left( \Ew \times \Sigma \right)_{ij} - \frac{2}{3} \left( \Ew \cdot \Sigma \right) g_{ij}  - J_{i \tau j},  \label{eq-1+3-bianchi-T-E}
 \end{split}
\end{equation}
\begin{equation}
\begin{split}
 & \lie_{\dtau} \Hw_{ij}  +  \curl \Ew_{ij} + \Hw_{ij}   \\
 = {} & \eta \Hw_{ij} - \frac{5}{2} \left( \Hw \times \Sigma \right)_{ij} - \frac{2}{3} \left( \Hw \cdot \Sigma \right) g_{ij}  - J^*_{i \tau j}, \label{eq-1+3-bianchi-T-H}
 \end{split}
\end{equation}
\end{subequations}
 where\footnote{For any symmetric $(0,2)$-tensors $A, B$ on $M$, $A \cdot B =A_{ij} B^{ij}$, $\curl A_{ij} = \frac{1}{2} \left( \epsilon_i{}^{\! pq}\nabla_q A_{pj} +\epsilon_j{}^{\! pq} \nabla_q A_{pi} \right)$,  
$(A \times B)_{ij} = \epsilon_i{}^{\! ab} \epsilon_j{}^{\! pq}  A_{ap} B_{bq} + \frac{1}{3} A \cdot B g_{ij} - \frac{1}{3} \tr A \tr B g_{ij}$. Note that, $\epsilon_{ijk}$ is the Kronecker symbol.}
 \begin{align*}
 J_{\beta \gamma \delta} &= \frac{1}{2} \left( \breve D_\gamma \breve R_{\delta \beta} - \breve D_{\delta} \breve R_{\gamma \beta}\right) \\
 &- \frac{1}{12} \left(\breve g_{\beta \delta} \breve D_\gamma ( \breve g^{\mu \nu} \breve R_{\mu \nu} ) - \breve g_{\beta \gamma}  \breve D_\delta (\breve g^{\mu \nu} \breve R_{\mu \nu} ) \right), \\
 J^*_{\beta \gamma \delta} &= \frac{1}{2} J_\beta{}^{\! \mu \nu} \epsilon_{\mu \nu \gamma \delta}. 
\end{align*}
In terms of $\Ew$ and $\Hw$, 
\al{Riem-Weyl}{
\R_{imjn} &= -\epsilon_{imp} \epsilon_{jnq} \Ew^{pq} + \frac{1}{2} (g_{ij} \R_{mn} + g_{mn} \R_{ij} - g_{mj} \R_{in} - g_{in} \R_{mj}) \nnb \\
&- \frac{1}{6} (g_{ij} g_{m n} - g_{in} g_{jn}) (\R_{pq} g^{pq} - \R_{\tau \tau}),
}
and \eqref{Gauss-Ricci-hat-k} becomes
\begin{align}\label{eq-Gauss-E} 
R_{ij} + 2 g_{ij} ={}  & \Ew_{ij} + ( 4 \eta - 2 \eta^2 ) g_{ij} + \Sigma_{ip} \Sigma_{j}^p + ( 1 - \eta ) \Sigma_{ij} \nnb \\
& + \frac{1}{2} \R_{ij} + \frac{1}{6}(2 \R_{\tau \tau} + \R) g_{ij}.
\end{align}
Now the Codazzi equation \eqref{Codazzi-curl-k} is equivalent to the system
\al{div-curl-k-hat}
{
(\dive \Sigma)_i &= 2 \nabla_i \eta - \R_{\tau i}, \nnb \\
(\curl \Sigma)_{ij} &= - \epsilon_{ij}{}^{\!m} \nabla_m \eta - \Hw_{ij}.  
}

In the end, we note an identity that is used throughout the paper \[\dtau \int_{M} f \di \mu_g = \int_M ( \dtau f - 3 \eta f ) \di \mu_g. \]

\subsection{The $1+3$ Einstein scalar field equations}
The Einstein scalar field equations \eqref{eq-Einstein-source}-\eqref{def-energy-Mom-kg} are reformulated as the $1+3$ decomposition \eqref{eq-evolution-1}-\eqref{eq-evolution-2} coupled with the constraints \eqref{Gauss-Ricci-trace-hat -k}-\eqref{Codazzi-div-k} and the scalar field \eqref{eq-kg} which can be decomposed in the following $1+3$ form,
\begin{equation}\label{eq-rescale-kg-1+3-0}
 \dtau^2 \phi +2\dtau \phi - \Delta \phi + m^2 t^2  \phi + 3 \eta \dtau \phi =0.
\end{equation}
The spacetime Ricci tensor $\R_{\alpha \beta}$ in the above geometric equations \eqref{eq-evolution-1}-\eqref{eq-evolution-2} is given by 
\begin{equation}\label{eq-ricci-phi}
\R_{\alpha \beta} = \breve{T}_{\alpha \beta} -  \frac{\tr_{\breve g} \breve{T}}{2} \breve g_{\alpha \beta} =\breve D_\alpha \phi \breve D_\beta \phi + \frac{m^2 t^2}{2}\phi^2 \bar g_{\alpha \beta}.
\end{equation}

Let $\xi_{ij} = \dtau g_{ij}$ and $\Delta_L$ be the Lichnerowicz Laplacian given by \[\Delta_L \xi_{ij} = \Delta \xi_{ij} + 2 R_i{}^m{}_j{}^l \xi_{ml} - R_i^l \xi_{j l} - R_j^l \xi_{i l}, \] then
\als{
\dtau R_{ij} = -\frac{1}{2} ( \Delta_L \xi_{ij} + \nabla_i \nabla_j (\xi^l_l) - \nabla_i (\dive \xi)_j - \nabla_j (\dive \xi)_i ).
}
From this, we obtain a second order equation for $\Sigma$\footnote{In fact, $\dtau^2 \Sigma_{ij} - \Delta \Sigma_{ij} + 2 \dtau \Sigma_{ij} - 2 \Sigma_{ij}$ is the principle part of $-\Box_{\bar g} \Sigma_{i j} - \dtau \Sigma_{i j} - 2 R_{ipjq} \Sigma^{p q}$, and the extra $-\dtau \Sigma$ is due to the rescaling of time function $\tau$.} by taking $\dtau$ derivative on \eqref{E:eq-trans-Sigma},
\al{wave-Sigma}{
& \quad \dtau^2 \Sigma_{ij} - \Delta \Sigma_{ij} + 2 \dtau \Sigma_{ij} - 2 \Sigma_{ij} \nnb \\
& =  \Delta \eta g_{i j} - 3 \nabla_i \nabla_j \eta  + (\dtau \Sigma, \Sigma, \eta) * (\Sigma, \eta)  +  (\Sigma, \eta) * (\Sigma, \eta) * (\Sigma, \eta)  \nnb \\
&- \dtau \R_{ij}(\phi)  + \nabla_i (\dtau \phi \nabla_j \phi) + \nabla_j (\dtau \phi \nabla_i \phi) + g_{ij} * \dtau \R_{\tau \tau} (\phi)  \nnb\\
%& + (\dtau \Sigma, \Sigma, \eta) * (\Sigma, \eta)  +  (\Sigma, \eta) * (\Sigma, \eta) * (\Sigma, \eta) \nnb \\
& + (1+ \eta) g_{ij} * \R_{\tau \tau} (\phi) +  \Sigma * (\R_{ m n}(\phi), \R_{\tau \tau}(\phi)),
}
where $\R_{ m n}(\phi)$, $\R_{\tau \tau}(\phi)$ are given by \eqref{eq-ricci-phi}, and $\ast$ is the contraction defined in subsection \ref{sec-contraction}.

\begin{comment}
It is worthy to note that $\Sigma$ satisfies the following wave equation
\al{wave-Sigma}
{& \dtau^2 \Sigma_{ij} - \Delta \Sigma_{ij} + 2 \dtau \Sigma_{ij} \nnb \\
={}&-3 \nabla_i \nabla_j \eta + \Delta \eta g_{i j} \nnb\\
&+ \dtau \R_{ij} (\phi) + \nabla_i (\dtau \phi \nabla_i \phi) + \nabla_j (\dtau \phi \nabla_i \phi) \nnb\\
& + \dtau \Sigma * (\Sigma, \eta) +  (\Sigma, \eta) * (\Sigma, \eta)  \nnb \\
& +  (\Sigma, \eta) * (\Sigma, \eta) * (\Sigma, \eta) \nnb \\
& +  (\Sigma, \eta) * (\Sigma, \eta) * (\R_{ m n} (\phi), \R_{\tau \tau} (\phi))  \pm \R_{m n} \pm \R_{\tau \tau}.
}
\end{comment}

%\subsection{Hyperbolic manifold}%\label{sec-sec-background}
\subsection{Sobolev norms}\label{sec-Sobolev}

For any $(p, q)$-tensor $\Psi \in T^p_q(M)$, we define  \[|\Psi|_g^2 : =  g^{i_1 j_1} \cdots g^{i_q j_q} g_{i^\prime_1 j^\prime_1} \cdots g_{i^\prime_p j^\prime_p}  \Psi_{i_1\cdots i_q}^{i^\prime_1 \cdots i^\prime_p} \Psi_{j_1 \cdots j_q}^{j^\prime_1 \cdots j^\prime_p}. \] 
In what follows, we also use $|\Psi|$ to denote $|\Psi|_g$ for simplicity. 

Let $H^p_k(M)$ be the Sobolev space of tensors with respect to the norm \[\|\Psi\|_{H^p_k} = \sum_{j=0}^k \left( \int_M |\nabla^j \Psi|^p \di \mu_g \right)^{\frac{1}{p}}. \]
Let $H_{0,k}^p(M)$ be the closure of the space of smooth tensors with compact support in $M$. 
Let $m$ be an integer and denote by $C_B^m(M)$ the space of function of class $C^m$ for which the norm $$\|\Psi\|_{C^m} = \sum_{j=0}^m \sup_{x \in M} |\nabla^j \Psi (x)|$$ is finite. 

For simplicity, we denote $H^2_k(M)$ by $H_k(M)$, and $\|\cdot\|_{H^2_k}$ by $\|\cdot\|_{H_k}$.
In addition, we denote $H_{N+2} (C^\infty(M))$ to be the completion of the space of smooth functions with respect to the $H_{N+2}$ norm, and $H_{0, N+2} (C^\infty(M))$ denotes the closure of the space  $C_0^\infty(M)$  in $H_{N+2}(C^\infty(M))$.

We recall the following Sobolev embedding theorem \cite{Hebey}\footnote{The theorem holds for any tensors because of Kato's inequality.}.
\begin{proposition}\label{prop-Sobolev}
Let $(M, g)$ be a smooth, complete Riemannian $n$-manifold with Ricci curvature bounded from below. Assume that for any $x \in M$, $$\text{Vol}_g (B_x(1)) >\kappa,$$ where $\kappa$ is a positive constant, and  $\text{Vol}_g (B_x(1))$ stands for the volume of unit ball centred at $x$, $B_x(1)$, with respect to $g$. 

Let $k >m$ be two integers.

\begin{itemize}
\item  For any $1 \leq q <n$ and $q\leq p$, $\frac{1}{p} \geq \frac{1}{q} - \frac{k-m}{n}$, $H_k^q (M) \subset H_m^p(M)$.
\item For any $q \geq 1$, if $\frac{1}{q} < \frac{k-m}{n}$, then $H_k^q (M) \subset C^m_B(M)$.
\end{itemize}
\end{proposition}

The next proposition which passes from Laplacian $\Delta$ to general covariant derivate $\nabla$ will be used throughout the paper.
\begin{proposition}\label{prop-elliptic-Delta}
Let $(M, g)$ satisfy the assumption of proposition \ref{prop-Sobolev}.
Fix an integer $N \geq 2$. Suppose\footnote{In the case that $M$ is 3-dimensional Riemanian manifold, the Riemman tensor here can be replaced by Ricc tensor.} \[ \|R_{imjn} \|_{L^\infty} \quad \text{and} \quad \|\nabla R_{ipjq} \|_{H_{N-1}}, \quad \quad N \geq 2 \] are bounded.
  %\[  \| R_{ipjq} \|_{H_{K}}, \quad \quad K \geq 2 \] are bounded. 
  For any $\Psi \in T^p_q (M)$ with compact support, we have
\[\|\Psi\|^2_{H_k} \lesssim   \sum_{l \leq k} \|\nabla^{\mathring{l}} \Delta^{ [\frac{l}{2}] } \Psi \|_{L^2}^2, \quad k \leq N+2. \]
%In addition, for any $\psi \in C^\infty(M)$ with compact support, \[\|\psi\|^2_{H_k} \lesssim  \sum_{l \leq k} \|\nabla^{\mathring{l}} \Delta^{ [\frac{l}{2}] } \psi \|_{L^2}^2, \quad k \leq K+2. \]
Here $\mathring{l}$ is an integer such that 
$\mathring{l} =
\begin{cases}
0, &\text{if} \,\, l \,\, \text{is even}, \\
1, &\text{if} \,\, l \,\, \text{is odd}.
\end{cases}$
%when $k$ is even, $\mathring{k}=0$; when $k$ is odd, $\mathring{k}=1$.
\end{proposition}
The proof is collected in the appendix, Section \ref{sec-comm}.

\subsection{More conventions}

\subsubsection{$M$-tensor}
Let $\Psi_{\alpha_1 \cdots  \alpha_l}$ be a $(0, l)$-tensor on $\mathcal{M}^4$ satisfying
\begin{equation*}
\p_t^\beta \Psi_{\alpha_1 \cdots \alpha_{i-1} \beta \alpha_{i+1} \cdots \alpha_l} =0, \quad \forall \, i \in \{1, \cdots, l\},
\end{equation*}
where $\alpha_0 $ and $\alpha_{l+1}$ are interpreted as being absent if $i=1$ or $l$.
We can restrict $\Psi$ on $M_t:=\{t\} \times M^3$, and naturally interpret it as a tensor field on the Riemannian manifold $(M^3, \, g)$. Such a tensor $\Psi$ will be also called an \emph{$M$-tensor}.

\subsubsection{Multi index} 
For notational convenience,  we use $\nabla_{I_l} \psi$ to denote the $l^{\text{th}}$ order covariant derivative $\nabla_{i_1} \cdots \nabla_{i_l}\psi$, where $I_l=\{ i_1 \cdots i_l\}$ is the multi index.

\subsubsection{Simplified conventions for derivatives}
When $m \neq 0$, we use the notation \[\D \phi \in \{ \dtau \phi, \nabla \phi, m t \phi\}. \] When $m=0$, it reduces to  \[D \phi \in \{ \dtau \phi, \nabla \phi \}. \]

\subsubsection{Contraction}\label{sec-contraction} Unless indicated otherwise, we use the metric $g_{i j}$ and its inverse to raise and lower indices. Throughout, we use $A*B$ to denote a linear combination of products of $A$ and $B$, with each product being a contraction (with respect to $g$) between the two tensors $A$ and $B$ on $M$. We also use the notation $(A, B) \ast (D, E)$ to denote all possible linear combinations of products of $A, \,B$ and $D, \, E$.
In the estimates, we will only employ the formula $\|A * B\| \lesssim \|A\| \|B\| $ which allows ourselves to ignore the detailed product structure at this point. $\|\cdot \|$ denotes some Sobolev norms associated to $g$. 

\subsubsection{Universal constant} We use $C$ to denote some universal constant which may vary from line to line.

\section{Global existence for massive Einstein-scalar field}\label{sec-EKG}
In this section, we prove the global existence theorem for the massive Einstein-scalar field equations. We shall prove the following propagation estimates.
\begin{theorem}[Main estimates for massive Einstein-scalar field]\label{Thm-main-ee-EKG}
Fix a constant $\delta \in \left( 0, \,  \frac{1}{6} \right)$ and an integer $N \geq 2$.
For $\varepsilon>0$  small enough and all initial data of the massive Einstein-scalar field equations and $I_{N+2} \in \mathbb{R}^{+}$ which satisfy
\al{ME-Data-1}{
&  \| \D \phi \|_{H_{N+1}}  (t_0)  +  \|\eta \|_{H_{N+2}}  (t_0)  +  \|\Sigma \|_{H_{N+1}}  (t_0) \nnb  \\
& \quad  + \|\Ew \|_{H_{N}} (t_0) +  \|\Hw \|_{H_{N}}  (t_0)  +  \|g_{ij} - \gamma_{ij} \|_{H_{N+1}}  (t_0) \leq \varepsilon I_{N+2},
}
%and \[\frac{7}{8} \gamma_{ij} \leq g_{ij} (t_0) \leq \frac{9}{8} \gamma_{ij} \quad \text{ as bilinear forms}, \] 
there is a constant $C(I_{N+2})$ depending only on $I_{N+2}$ (in particular, not on $\varepsilon$ and $t_0$) such that for all $t\geq t_0$, we have
\al{energy-estimate}
{
& t^{1-C(I_{N+2})\varepsilon} ( \|\Ew \|_{H_{N}} +  \|\Hw \|_{H_{N}}  +  \|\eta \|_{H_{N+1}}   ) (t)  + t^{\frac{1}{2}} \| \D \phi \|_{H_{N+1}}  (t) \nnb \\
&\,\,    +  t^{1-\delta}  \|\Sigma \|_{H_{N+1}} (t)  + \|\eta \|_{H_{N+2}} (t) + \|g_{ij} - \gamma_{ij} \|_{H_{N+1}}(t)  \leq \varepsilon C(I_{N+2}),
}
and  
\als{
  t^{1-\delta} \| R_{ij} + 2 g_{ij } \|_{H_{N}} (t) &\leq \varepsilon C(I_{N+2}).
   % t^{1-\delta} \| R_{imjn} + g_{ij } g_{m n} - g_{in} g_{jm} \|_{H_{N}} &  \leq \varepsilon C(I_{N+2}).
  }
%\[ t^{1-\delta} (  \|\dtau \eta \|_{H_{N+1}}   +  \|\dtau \Sigma \|_{H_{N}}   ) (t) \leq \varepsilon C(I_{N+2}). \]
The subscript $N+2$ of $I_{N+2}$ denotes the number of derivatives involved in the norms.
\end{theorem}
\begin{remark}\label{rk-data}
We remark that the data bound \eqref{E:ME-Data-1} is equivalent to \eqref{intro-initial-data} in Theorem \ref{thm-global-existence}, in view of the Gauss--Codazzi equations \eqref{eq-Gauss-E}--\eqref{E:div-curl-k-hat}. Therefore, the local existence theorem \ref{thm-local-existence} is available when we intend to prove Theorem \ref{Thm-main-ee-EKG} through the continuity argument. It holds similarly in the massless case.
\end{remark}

\subsection{Bootstrap assumptions}\label{sec-bt-ekg}
Recall the fixed numbers $0<\delta < \frac{1}{6}$ and $N \geq 2$.
We start with the following weak assumptions: Let $\Lambda$ be a large constant to be determined, and
\begin{subequations}
\begin{align}
t ( \|\Ew \|_{H_{N}} +  \|\Hw \|_{H_{N}} ) & \leq \varepsilon \Lambda t^{\delta}, \label{BT-E-H} \\
t ( \|\eta \|_{H_{N+1}} +  \|\Sigma \|_{H_{N+1}} )  & \leq  \varepsilon  \Lambda t^{\delta},  \label{BT-k-N-L-infty} \\
\|\Ew \|_{C^0} + \|\eta \|_{C^0} + \| \Sigma \|_{C^0} + \| g_{ij} - \gamma_{i j} \|_{C^1} + \|\D \phi \|_{C^0}  & \leq \varepsilon \Lambda, \label{BT-g} \\
%\| \nabla_{I_{N+2}} \eta \|_{L^2} & \leq  \varepsilon  \Lambda, \label{BT-eta-top} \\
  t^{\frac{1}{2}}  \|\D \phi \|_{H_{N+1}}   & \leq \varepsilon \Lambda. \label{BT-KG}
\end{align}
\end{subequations}
We will improve these bootstrap assumptions by showing that \eqref{BT-E-H}--\eqref{BT-KG} implies the same inequalities hold with the constant $\Lambda$ replaced by $\frac{1}{2} \Lambda$.
%In this process, we remark that the bootstrap assumption \eqref{BT-eta-top} is in fact redundant. It is added only to guarantee the local existence of the solution.

Since $\Lambda$ is independent of $\epsilon$, then for $\varepsilon>0$ sufficiently small, \eqref{BT-g} implies that $g$ and $\gamma$ are equivalent as bilinear forms, and $g$ is close to $\gamma$ in $C^1$ norm.
%\begin{align}\label{kbt-g} 
 %\frac{3}{4} \gamma_{ij} < g_{ij} < \frac{5}{4} \gamma_{ij}
%\end{align}
%as bilinear forms.
Since $M_t=\{t\} \times M^3$ is diffeomorphic to $\mathbb{R}^3$, and $g$ is close to $\gamma$, the spatial manifold $(M_t, g)$ has infinite injectivity radius. 
Moreover, under the bootstrap assumptions \eqref{BT-E-H}--\eqref{BT-KG}, we know from \eqref{eq-Gauss-E} and \eqref{eq-ricci-phi} that the Ricci tensor of $g_{ij}$ is bounded from below
\be\label{Ric-low-bound}
R_{ij} > -\frac{3}{2} g_{ij}.
\ee
Therefore, by Proposition \ref{prop-Sobolev}, under the bootstrap assumptions \eqref{BT-E-H}--\eqref{BT-KG}, the Sobolev inequalities on $(M^3, \, g)$ hold. Hence, we have
\al{pre-estimates}{
&  \|\Ew  \|_{C^{N-2}}+  \|\Hw  \|_{C^{N-2}}+ \|\Ew \|_{H^4_{N-1}} + \|\Hw \|_{H^4_{N-1}}  \lesssim \varepsilon \Lambda t^{\delta-1}, \nnb \\
&   \|\eta  \|_{C^{N-1}}+  \|\Sigma  \|_{C^{N-1}}+ \|\eta \|_{H^4_{N}} + \|\Sigma \|_{H^4_{N}}   \lesssim \varepsilon \Lambda t^{\delta-1}, \nnb \\
& \|\D \phi \|_{C^{N-1}} + \|\D \phi \|_{H^4_{N}} \lesssim \varepsilon \Lambda t^{-\frac{1}{2}} .
}

Furthermore, making use of the evolution equations \eqref{eq-evolution-1} and \eqref{BT-k-N-L-infty}, we deduce
\be\label{pre-g-gamma}
\|g_{ij} - \gamma_{i j}\|_{H_{N+1}} \lesssim \varepsilon \Lambda < \frac{1}{2},
\ee 
as $\varepsilon$ is small enough. Then, the density theorem follows from Proposition \ref{pro-density}.
\begin{corollary}\label{coro-density} 
Under the bootstrap assumptions \eqref{BT-E-H}--\eqref{BT-KG} 
\als{
H_{0, k}(M) &= H_{k}(M), \quad k \leq N + 1,\\
H_{0, l} (C^\infty(M)) &= H_{l} (C^\infty(M)), \quad l \leq N + 2,
} 
where $H_{0, k}(M)$, $H_{k}(M)$, $H_{l} (C^\infty(M))$ and $H_{0, l} (C^\infty(M))$ are defined in Section \ref{sec-Sobolev}.
\end{corollary}

Under the bootstrap assumptions \eqref{BT-E-H}--\eqref{BT-KG}, we know from \eqref{E:Gauss-Riem-hat-k}, \eqref{E:Riem-Weyl}, \eqref{eq-ricci-phi} and \eqref{E:pre-estimates} that the Riemanian tensor of $g$ is bounded as follows
\al{Reim-bound}{
 \|R_{imjn}\|_{C^{N-2}} < C, \quad \|\nabla R_{imjn} \|_{H_{N-1}} < C, \quad N \geq 2.
}
%where $C_1, \, C_2 >0$ are both universal constants.
As a consequence of Corollary \ref{coro-density} and the estimates \eqref{E:Reim-bound}, we can improve Proposition  \ref{prop-elliptic-Delta} as follow.
\begin{corollary}\label{prop-elliptic-Delta-1}
Under the bootstrap assumptions \eqref{BT-E-H}--\eqref{BT-KG}, the conclusion of Proposition \ref{prop-elliptic-Delta} holds for $\Psi \in H_{k+1}(M)$, $\psi \in H_{k+2} (C^\infty(M))$, $k \leq N$ with $N\geq 2$.
\end{corollary}

\subsection{Energy estimates for the Einstein Klein--Gordon equations}
Define the energy for the Klein--Gordon field
\begin{align}\label{def-energy-kg}
E_{[k+1]} (\phi, t) & =  \int_{M_t} \left( |\nabla^{\mathring{k}} \Delta^{ [\frac{k}{2}] } \dtau \phi|^2 +  |\nabla^{ k^\prime } \Delta^{ [\frac{k+1}{2}] }  \phi|^2 +  m^2 t^2 |\nabla^{\mathring{k}} \Delta^{ [\frac{k}{2}] } \phi|^2 \right) \di \mu_g,
\end{align}
where $2 [\frac{k}{2}] + \mathring{k} =k$, $2 [\frac{k+1}{2}] + k^\prime =k+1$. Namely, 
\begin{equation}\label{def-k-prime} 
\mathring{k} =
\begin{cases}
0, &\text{if} \,\, k \,\, \text{is even}, \\
1, &\text{if} \,\, k \,\, \text{is odd};
\end{cases} 
\qquad
k^\prime =
\begin{cases}
1, &\text{if} \,\, k \,\, \text{is even}, \\
0, &\text{if} \,\, k \,\, \text{is odd}.
\end{cases}
\end{equation}

Define the $k^{\text{th}}$-order energy norm for the Weyl tensor
\begin{equation}\label{def-energy-l-homo-Weyl}
E_{[k]} (\W, t) =  \int_{M_t}  \left( |\nabla^{\mathring{k}} \Delta^{ [\frac{k}{2}] } \Ew|^2 + |\nabla^{\mathring{k}} \Delta^{ [\frac{k}{2}] } \Hw |^2  \right) \di \mu_g.
\end{equation}

We drop the bracket to denote energy norm up to $l$ orders,
\begin{equation}\label{def-energy-sum-inhomo}
E_l (\phi, t) = \sum_{k \leq l} E_{[k]}(\phi, t), \quad E_l (\W, t) = \sum_{k \leq l} E_{[k]}(\W, t).
\end{equation}

\subsubsection{Energy estimates for the Klein--Gordon equation}\label{sec-ee-kg-1+3}
The high order Klein--Gordon equation takes the form of
\al{eq-kg-N-l-1+3-general-simply}
 { &\lie_{\dtau} \nabla^{\mathring{l}} \Delta^{ [\frac{l}{2}] } \dtau \phi + 2 \nabla^{\mathring{l}} \Delta^{ [\frac{l}{2}] }  \dtau \phi -  \nabla^{\mathring{l}} \Delta^{ [\frac{l}{2}] +1 }  \phi + m^2 t^2  \nabla^{\mathring{l}} \Delta^{ [\frac{l}{2}] }  \phi \nnb \\
 = {} & \nabla^{\mathring{l}} \Delta^{ [\frac{l}{2}] } ( \eta * \dtau \phi) + \sum_{a+1 +b  = l}   \left(   \nabla_{I_a} \Sigma, \, \nabla_{I_a} \eta \right) * \nabla_{I_{b}} \nabla \dtau \phi.
   }

We use the modified energy \cite{Wang-J-EKG-2019} instead of $E_{[k+1]}$ 
\begin{equation}\label{def-tilde-energy-norm-l-kg}
\tilde{E}_{[k+1]}(\phi, t) =  E_{[k+1]} (\phi, t) +  \int_{M_t}  3\nabla^{\mathring{k}} \Delta^{ [\frac{k}{2}] } \dtau \phi  \cdot \nabla^{\mathring{k}} \Delta^{ [\frac{k}{2}] } \phi \, \di \mu_g,
\end{equation}
which will be helpful in the proof of energy estimates.
For $3 <  m t$ (or $t > 3 m^{-1}$), the two energy norms are equivalent, \[\frac{1}{2} E_{[k+1]} (\phi, t) <  \tilde E_{[k+1]} (\phi, t)< \frac{3}{2} E_{[k+1]} (\phi, t). \] For the moment, we  require $t\geq t_0 > 3 m^{-1}$. Nevertheless, we can eventually drop this requirement on initial time, since by the standard theory of wave equations, finite time existence with small data is always true for the Klein--Gordon equation.

\begin{proposition}\label{prop-ee-kg}
We have the following estimate for the Klein--Gordon field, 
\be\label{energy-B-kg} 
 \|\D \phi\|_{H_{N+1}} \lesssim \varepsilon I_{N+2} t^{-\frac{1}{2}}, \quad N \geq 2,
 \ee
 provided the bootstrap assumptions \eqref{BT-E-H}--\eqref{BT-KG}.
\end{proposition}

\begin{proof}
The proof is a slight modification of the one in \cite{Wang-J-EKG-2019}. Here we use $\nabla^{\mathring{l}} \Delta^{ [\frac{l}{2}] }$  instead of general $\nabla_{I_l}$ to simplify the calculations.

Multiplying $2 \nabla^{\mathring{l}} \Delta^{ [\frac{l}{2}] } \dtau \phi$ on both sides of \eqref{E:eq-kg-N-l-1+3-general-simply}, and applying Lemma \ref{lemma-commuting-application}, we have
\als{
& \dtau |\nabla^{\mathring{l}} \Delta^{ [\frac{l}{2}] } \dtau \phi|^2 +  \dtau (\nabla^{ l^\prime} \Delta^{ [\frac{l+1}{2}]  }  \phi  \nabla^{ l^\prime} \Delta^{ [\frac{l+1}{2}] }  \phi) + \dtau ( m^2 t^2 | \nabla^{\mathring{l}} \Delta^{ [\frac{l}{2}] }  \phi|^2 )  \nnb \\
& - 2 \nabla (\nabla^{ l^\prime } \Delta^{ [\frac{l+1}{2}] }  \phi  \nabla^{\mathring{l}} \Delta^{ [\frac{l}{2}] } \dtau \phi)  + 4 |\nabla^{\mathring{l}} \Delta^{ [\frac{l}{2}] }  \dtau \phi|^2 - 2 m^2 t^2 | \nabla^{\mathring{l}} \Delta^{ [\frac{l}{2}] }  \phi|^2 \nnb \\
% = {} & \nabla^{I_l} ( \eta * \dtau \phi)  \nabla_{I_l} \dtau \phi + (\Sigma + \eta) * \nabla_{I_l} \D \phi  * \nabla_{I_l} \D \phi \\
% &+ \sum_{a+1 +b  = l}   \left(   \nabla_{I_a} \Sigma +   \nabla_{I_a} \eta \right)* \nabla_{I_{b}} \nabla \dtau \phi \nabla_{I_l} \dtau \phi \\
% & + \sum_{a+1 +b  = l+1}  \left(   \nabla_{I_a} \Sigma +   \nabla_{I_a} \eta \right) * \nabla_{I_{b}} \nabla  \phi \nabla_{I_{l+1}}  \phi \\
% & + \sum_{a+1 +b  = l} m^2 t^2  \left(   \nabla_{I_a} \Sigma +   \nabla_{I_a} \eta \right) * \nabla_{I_{b}} \nabla  \phi \nabla_{I_{l}}  \phi \\
  = {} &  \sum_{a +b  = l}  \left(   \nabla_{I_a} \Sigma, \,   \nabla_{I_a} \eta \right) * \nabla_{I_{b}} \D \phi \nabla^{\mathring{l}} \Delta^{ [\frac{l}{2}] } \D \phi,
   }
Note that
\als{
& 4 |\nabla^{\mathring{l}} \Delta^{ [\frac{l}{2}] }  \dtau \phi|^2 - 2 m^2 t^2 | \nabla^{\mathring{l}} \Delta^{ [\frac{l}{2}] }  \phi|^2 \\
={}&  |\nabla^{\mathring{l}} \Delta^{ [\frac{l}{2}] }  \dtau \phi|^2 + 3 \dtau( \nabla^{\mathring{l}} \Delta^{ [\frac{l}{2}] }  \dtau \phi  \nabla^{\mathring{l}} \Delta^{ [\frac{l}{2}] }   \phi)  - 2 m^2 t^2 | \nabla^{\mathring{l}} \Delta^{ [\frac{l}{2}] }  \phi|^2 \\
& - 3 \dtau \nabla^{\mathring{l}} \Delta^{ [\frac{l}{2}] }  \dtau \phi \cdot \nabla^{\mathring{l}} \Delta^{ [\frac{l}{2}] }   \phi  +  \sum_{a+1 +b  = l}  \left(   \nabla_{I_a} \Sigma, \,   \nabla_{I_a} \eta \right)* \nabla_{I_{b}}  \nabla \phi \nabla^{\mathring{l}} \Delta^{ [\frac{l}{2}] }  \dtau \phi,
}
and moreover
\als{
& - 3 \dtau \nabla^{\mathring{l}} \Delta^{ [\frac{l}{2}] }  \dtau \phi \cdot \nabla^{\mathring{l}} \Delta^{ [\frac{l}{2}] }   \phi    \\
={}& - 3 \nabla^{\mathring{l}} \Delta^{ [\frac{l}{2}] }  \dtau^2 \phi \cdot \nabla^{\mathring{l}} \Delta^{ [\frac{l}{2}] }   \phi +  \sum_{a+1 +b  = l}   \left(   \nabla_{I_a} \Sigma,\,   \nabla_{I_a} \eta \right) * \nabla_{I_{b}}  \nabla \phi \nabla^{\mathring{l}} \Delta^{ [\frac{l}{2}] }  \dtau \phi \\
={} &  3 \left(  2 \nabla^{\mathring{l}} \Delta^{ [\frac{l}{2}] }  \dtau \phi -  \nabla^{\mathring{l}} \Delta^{ [\frac{l}{2}] +1 }  \phi + m^2 t^2  \nabla^{\mathring{l}} \Delta^{ [\frac{l}{2}] }  \phi\right) \cdot \nabla^{\mathring{l}} \Delta^{ [\frac{l}{2}] }   \phi + \mathcal{N}_l,
}
where we applied the Klein--Gordon equation in the last step and 
\als{
\mathcal{N}_l = {} & \nabla^{\mathring{l}} \Delta^{ [\frac{l}{2}] } ( \eta * \dtau \phi)  \nabla^{\mathring{l}} \Delta^{ [\frac{l}{2}] }   \phi+ \sum_{a+1 +b  = l}   \left(   \nabla_{I_a} \Sigma,\,   \nabla_{I_a} \eta \right) * \nabla_{I_{b}} \nabla \dtau \phi \nabla^{\mathring{l}} \Delta^{ [\frac{l}{2}] }   \phi \\
& +  \sum_{a+1 +b  = l}   \left(   \nabla_{I_a} \Sigma,\,   \nabla_{I_a} \eta \right) * \nabla_{I_{b}}  \nabla \phi \nabla^{\mathring{l}} \Delta^{ [\frac{l}{2}] }  \dtau \phi.
}  Taking the identity
\[ - 3 \nabla^{\mathring{l}} \Delta^{ [\frac{l}{2}] +1 }  \phi \cdot \nabla^{\mathring{l}} \Delta^{ [\frac{l}{2}] }   \phi = 3 | \nabla^{ l^\prime} \Delta^{ [\frac{l+1}{2}] }  \phi |^2 - 3 \nabla  (\nabla^{ l^\prime}  \Delta^{ [\frac{l+1}{2}] }  \phi \nabla^{ \mathring{l} } \Delta^{ [\frac{l}{2}] }  \phi ), \]
into account as well, we obtain
\als{
& 4 |\nabla^{\mathring{l}} \Delta^{ [\frac{l}{2}] }  \dtau \phi|^2 - 2 m^2 t^2 | \nabla^{\mathring{l}} \Delta^{ [\frac{l}{2}] }  \phi|^2 \\
% ={}&  |\nabla^{\mathring{l}} \Delta^{ [\frac{l}{2}] }  \dtau \phi|^2 + 3 \dtau( \nabla^{\mathring{l}} \Delta^{ [\frac{l}{2}] }  \dtau \phi  \nabla^{\mathring{l}} \Delta^{ [\frac{l}{2}] }   \phi)  - 2 m^2 t^2 | \nabla^{\mathring{l}} \Delta^{ [\frac{l}{2}] }  \phi|^2 \\
% &+ 3 \left(  2 \nabla^{\mathring{l}} \Delta^{ [\frac{l}{2}] }  \dtau \phi -  \nabla^{\mathring{l}} \Delta^{ [\frac{l}{2}] +1 }  \phi + m^2 t^2  \nabla^{\mathring{l}} \Delta^{ [\frac{l}{2}] }  \phi\right) \cdot \nabla^{\mathring{l}} \Delta^{ [\frac{l}{2}] }   \phi  \\
% &+  \sum_{a+1 +b  = l}   \left(   \nabla_{I_a} \Sigma +   \nabla_{I_a} \eta \right) * \nabla_{I_{b}}  \nabla \phi \nabla^{\mathring{l}} \Delta^{ [\frac{l}{2}] }  \dtau \phi \\
% &+ \nabla^{\mathring{l}} \Delta^{ [\frac{l}{2}] } ( \eta * \dtau \phi)  \nabla^{\mathring{l}} \Delta^{ [\frac{l}{2}] }   \phi+ \sum_{a+1 +b  = l}   \left(   \nabla_{I_a} \Sigma +   \nabla_{I_a} \eta \right) * \nabla_{I_{b}} \nabla \dtau \phi \nabla^{\mathring{l}} \Delta^{ [\frac{l}{2}] }   \phi \\
={}& 3 \dtau( \nabla^{\mathring{l}} \Delta^{ [\frac{l}{2}] }  \dtau \phi  \nabla^{\mathring{l}} \Delta^{ [\frac{l}{2}] }   \phi) +  |\nabla^{\mathring{l}} \Delta^{ [\frac{l}{2}] }  \dtau \phi|^2 + m^2 t^2 | \nabla^{\mathring{l}} \Delta^{ [\frac{l}{2}] }  \phi|^2+ 3 | \nabla^{ l^\prime} \Delta^{ [\frac{l+1}{2}] }  \phi |^2 \\
& + 6  \nabla^{\mathring{l}} \Delta^{ [\frac{l}{2}] }  \dtau \phi \nabla^{\mathring{l}} \Delta^{ [\frac{l}{2}] }  \phi - 3 \nabla  (\nabla^{ l^\prime}  \Delta^{ [\frac{l+1}{2}] }  \phi \nabla^{ \mathring{l} } \Delta^{ [\frac{l}{2}] }  \phi ) + \mathcal{N}_l. % \\
%&+  \sum_{a+1 +b  = l}   \left(   \nabla_{I_a} \Sigma +   \nabla_{I_a} \eta \right) * \nabla_{I_{b}}  \nabla \phi \nabla^{ \mathring{l} } \Delta^{ [\frac{l}{2}] }  \dtau \phi  \\
%&+ \sum_{a +b  = l}  \left(   \nabla_{I_a} \Sigma +   \nabla_{I_a} \eta \right) * \nabla_{I_{b}} \dtau \phi \nabla^{ \mathring{l} } \Delta^{ [\frac{l}{2}] }  \phi.
}
As a result, we obtain the energy identity,
\al{energy-id-l-kg-rescale}
{&\dtau  \tilde E_{[l+1]} (\phi, t)   + \tilde E_{[l+1]} (\phi, t)  +  \int_{M_t} 2  |\nabla^{ l^\prime} \Delta^{ [\frac{l+1}{2}] }  \phi|^2\, \di \mu_g \nnb \\
  = &  \int_{M_t} 3 \nabla^{\mathring{l}} \Delta^{ [\frac{l}{2}] } \phi \nabla^{\mathring{l}} \Delta^{ [\frac{l}{2}] }  \dtau \phi \, \di \mu_g \nnb \\
  %+  \sum_{a+1 +b  = l}  \left(   \nabla_{I_a} \Sigma +   \nabla_{I_a} \eta \right) * \nabla_{I_{b}}  \nabla \phi \nabla^{I_l}  \dtau \phi \nnb \\
  & +   \int_{M_t} \sum_{k \leq l} \nabla^{I_k} \left( (\eta,\, \Sigma_{ij}) * \D \phi \right) * \left( \nabla^{ \mathring{l} } \Delta^{ [\frac{l}{2}] } \D \phi, \, \nabla^{ \mathring{l} } \Delta^{ [\frac{l}{2}] }  \phi\right)  \di \mu_g.
}
Noting that \[ \| \nabla^{ \mathring{l} } \Delta^{ [\frac{l}{2}] } \phi \|_{L^2} \lesssim t^{-1} \tilde E^{\frac{1}{2}}_{[l+1]}  (\phi, t), \] thus
\[\p_t ( t  \tilde E_{[l+1]} (\phi, t) ) \lesssim  t^{-2}  \cdot t  \tilde E_{[l+1]}  (\phi, t) + \varepsilon \Lambda t^{- 2 + \delta} \cdot t^{\frac{1}{2}}  \tilde E^{\frac{1}{2}}_{[l+1]}  (\phi, t), \quad l \leq N+1. \]
Using the Gr\"{o}nwall's inequality and the equivalence between $E_{[l+1]}$ and $\tE_{[l+1]}$, we have \[  E_{N +2} (\phi, t) \lesssim \varepsilon^2 I^2_{N+2} t^{-1}. \] Now the conclusion follows from Corollary \ref{prop-elliptic-Delta-1}.
\end{proof}

\subsubsection{Energy estimates for the $1+3$ Bianchi equations}\label{sec-ee-Bianchi-1+3}
According to the $1+3$ Bianchi equations \eqref{eq-1+3-bianchi-T-E}--\eqref{eq-1+3-bianchi-T-H}, we obtain the high order equations 
\als{
 & \lie_{\dtau} \nabla^{\mathring{k}}  \Delta^{ [\frac{k}{2}] } \Ew_{ij} - \nabla^{\mathring{k}}  \Delta^{ [\frac{k}{2}] } \curl \Hw_{ij} +\nabla^{\mathring{k}}  \Delta^{ [\frac{k}{2}] } \Ew_{ij} \\
   = {} & \nabla_{I_k} \left( (\eta, \Sigma) * \Ew_{ij} \right)  - \nabla^{\mathring{k}}  \Delta^{ [\frac{k}{2}] }  J_{i \tau j},  \\
 & \lie_{\dtau}\nabla^{\mathring{k}}  \Delta^{ [\frac{k}{2}] }  \Hw_{ij}  + \nabla^{\mathring{k}}  \Delta^{ [\frac{k}{2}] }   \curl \Ew_{ij} + \nabla^{\mathring{k}}  \Delta^{ [\frac{k}{2}] }  \Hw_{ij}   \\
 = {} &\nabla_{I_k} \left( (\eta, \Sigma) * \Hw_{ij} \right)   - \nabla^{\mathring{k}}  \Delta^{ [\frac{k}{2}] }  J^*_{i \tau j}.
}

In order to derive high order energy estimates for the $1+3$ Bianchi equations, we have to appeal to an identity,  which is a counterpart of  \cite[Lemma 3.14]{Wang-J-EKG-2019}. Furthermore, using the derivatives $\nabla^{\mathring{k}} \Delta^{ [\frac{k}{2}] }$, we are free to proceed to arbitrary order derivatives, while the proof for $k > 3$ in  \cite[Lemma 3.14]{Wang-J-EKG-2019} would be highly complicated.

\begin{lemma}\label{lemma-div-curl}
 Let  $k \in \mathbb{Z}$, $k \geq 0$, and $H_{ij}, \, E_{ij}$ be any two symmetric $(0, 2)$-tensors on $M$. Then, the following identity holds
\als{
&   \nabla^{\mathring{k}} \Delta^{ [\frac{k}{2}] } \curl H \cdot \nabla^{\mathring{k}} \Delta^{ [\frac{k}{2}] } E  -  \nabla^{\mathring{k}} \Delta^{ [\frac{k}{2}] }  \curl E  \cdot \nabla^{\mathring{k}} \Delta^{ [\frac{k}{2}] } H  \\
=& \sum_{l \leq k}  \nabla_q    \left(  \epsilon_i{}^{\! pq}  \nabla^{\mathring{l}} \Delta^{ [\frac{l}{2}] } H_{p j}  * \nabla^{\mathring{l}} \Delta^{ [\frac{l}{2}] } E^{i j}  \right) \\
&+  \nabla_{I_{l-1}} ( O_{imjn}  \ast H) *  \nabla^{\mathring{l}} \Delta^{ [\frac{l}{2}] } E +  \nabla_{I_{l-1}} ( O_{imjn}  \ast E) *  \nabla^{\mathring{l}} \Delta^{ [\frac{l}{2}] } H, %\label{wedge-curl-div-general}
}
where $O_{imjn} = R_{imjn} + \frac{1}{2} ( g \odot g )_{imjn}$ is defined to be the error term of the Rieman curvature with respect to its principle part, and the last line vanishes when $l=0$. %By the Gauss equation \eqref{E:Gauss-Riem-hat-k}, $ O_{imjn} = ( \eta -  \frac{\eta^2}{2} ) ( g \odot g )_{imjn}   + ( 1- \eta ) ( g \odot \Sigma )_{imjn} - \frac{1}{2} ( \Sigma \odot \Sigma )_{imjn}   + \R_{imjn}$, and $\R_{imjn}$ is given in \eqref{E:Riem-Weyl}.
\end{lemma}

The proof of Lemma \ref{lemma-div-curl}  is given in Appendix \ref{sec-id-Bianchi}.

\begin{proposition}\label{prop-energy-estimate-Bianchi}

Under the bootstrap assumptions \eqref{BT-E-H}--\eqref{BT-KG}, there is a constant $C(I_{N+2})$ depending only on $I_{N+2}$, such that for $N \geq 2$, 
\[t^2 ( \|\Ew\|^2_{H_{N-1}} + \|\Hw\|^2_{H_{N-1}})  \lesssim  \varepsilon^2 I_{N+1}^2 +  \varepsilon^3 \Lambda^3  \]
and
\[ t^2  ( \|\Ew\|^2_{H_{N}} + \|\Hw\|^2_{H_{N}})  \lesssim \left(  \varepsilon^2 I_{N+2}^2 +  \varepsilon^3 \Lambda^3 \right)t^{ \varepsilon C( I_{N+2}) }.  \]
\end{proposition} 

\begin{proof}
Here we only sketch the proof. For more details, please refer to \cite{Wang-J-EKG-2019}.
 
With the help of Lemma \ref{lemma-div-curl}, we derive the energy identity,
\als%{energy-id-l-weyl-rescale}
{&\dtau  E_{[l+1]} (\W, t)   + 2 E_{[l+1]} (\W, t)    \nnb \\
  = &  \sum_{k \leq l}  \int_{M_t}  \nabla_{ I_{k}} \left( \eta  * \W + \Sigma  * \W  \right) *  \nabla^{\mathring{l}} \Delta^{ [\frac{l}{2}] }\W \, \di \mu_g +  \int_{M_t}  S_l \, \di \mu_g \nnb \\
  & +  \sum_{k \leq l}  \int_{M_t}  \nabla_{ I_{k-1}} \left( \Ew * \W + (\Sigma, \eta)  * (\Sigma, \eta) + \D \phi * \D \phi \right) * \nabla^{\mathring{k}} \Delta^{ [\frac{k}{2}] } \W\, \di \mu_g,
}
where the last line is due to $\sum_{k \leq l} \nabla_{I_{k-1}} \left( O_{imjn} \ast \W \right) * \nabla^{\mathring{k}} \Delta^{ [\frac{k}{2}] } \W$ in Lemma \ref{lemma-div-curl}. Here
\[ O_{imjn} = ( \eta -  \frac{\eta^2}{2} ) ( g \odot g )_{imjn}   + ( 1- \eta ) ( g \odot \Sigma )_{imjn} - \frac{1}{2} ( \Sigma \odot \Sigma )_{imjn}   + \R_{imjn},\] with $\R_{imjn}$ given in \eqref{E:Riem-Weyl}, by the Gauss equation \eqref{E:Gauss-Riem-hat-k}. As in \cite{Wang-J-EKG-2019}, the source term $S_l$ is given by
\al{coupling-cancellation}{
S_l & = -\nabla^{\mathring{l}}  \Delta^{ [\frac{l}{2}] } \left( J_{i \tau j}\right) \nabla^{\mathring{l}}  \Delta^{ [\frac{l}{2}] } \Ew^{ij} - \nabla^{\mathring{l}}  \Delta^{ [\frac{l}{2}] } \left(  J^*_{i \tau j}\right) \nabla^{\mathring{l}}  \Delta^{ [\frac{l}{2}] } \Hw^{ij} \nnb\\
&= \nabla^{\mathring{l}}  \Delta^{ [\frac{l}{2}] } \left( \nabla \dtau \phi * \nabla \phi +   \nabla^2 \phi * \dtau \phi \right) * \nabla^{\mathring{l}}  \Delta^{ [\frac{l}{2}] } \Ew \nnb \\
&\quad + \nabla^{\mathring{l}}  \Delta^{ [\frac{l}{2}] } \left( \left(k_i^j * \nabla \phi * \nabla \phi + \Sigma* (\dtau \phi)^2 \right) \right) * \nabla^{\mathring{l}}  \Delta^{ [\frac{l}{2}] } \Ew \nnb \\
&\quad + \nabla^{\mathring{l}}  \Delta^{ [\frac{l}{2}] } \left(  \left( \nabla^2 \phi + \Sigma* \dtau \phi \right) * \nabla \phi \right) * \nabla^{\mathring{l}}  \Delta^{ [\frac{l}{2}] } \Hw.
}
In fact, a cancellation is hidden in the calculation of $S_l$ \eqref{E:coupling-cancellation}. Namely, there is no  $\dtau^2 \phi$ in the final formula of  \eqref{E:coupling-cancellation}.
To estimate the nonlinear terms,  we typically take $ \nabla^{\mathring{l}}  \Delta^{ [\frac{l}{2}] } \left(  \nabla^2 \phi * \dtau \phi \right) * \nabla^{\mathring{l}}  \Delta^{ [\frac{l}{2}] } \Ew$ in $S_l$ for example. Other terms can be estimated in a more straightforward way. 

We begin with the lower order derivatives, $l \leq N-1$.
Note that, 
\al{bl-W-low}{
 \|\nabla_{I_a} \nabla^2 \phi\|_{L^2} & \lesssim  t^{-1} \|\nabla_{I_a} \nabla^2 ( m t \phi) \|_{L^2} \nnb \\
 &<  t^{-1} \| \D \phi \|_{H_{N+1}} \lesssim \varepsilon \Lambda t^{-\frac{3}{2}},
 } 
 for $a \leq N-1$.
Due to the extra factor $t^{-1}$, we can estimate directly
 \als{
& \| \nabla^{\mathring{l}}  \Delta^{ [\frac{l}{2}] } \left(  \nabla^2 \phi * \dtau \phi \right) * \nabla^{\mathring{l}}  \Delta^{ [\frac{l}{2}] } \Ew\|_{L^2} \\
 \lesssim {} &  \varepsilon^2 \Lambda^2 t^{-2} \| \Ew \|_{H_{l}}\lesssim  \varepsilon^3 \Lambda^3 t^{-3+\delta}, \quad l \leq N-1. 
 }
 
For the top order $l=N$, the estimate \eqref{E:bl-W-low} is no longer valid due to the regularity. Instead, it should be replaced by
 \be\label{bl-W-top}
 \|\nabla_{I_N} \nabla^2 \phi\|_{L^2}  < \| \nabla \phi \|_{H_{N+1}} \lesssim \varepsilon I_{N+2}  t^{-\frac{1}{2}},
\ee
due to \eqref{energy-B-kg}.
Combining \eqref{E:bl-W-low}--\eqref{bl-W-top}, we obtain
 \als{
& \| \nabla^{\mathring{N}}  \Delta^{ [\frac{N}{2}] } \left(  \nabla^2 \phi * \dtau \phi \right) * \nabla^{\mathring{N}}  \Delta^{ [\frac{N}{2}] } \Ew\|_{L^2} \\
 \lesssim {} & \varepsilon^3 \Lambda^3 t^{-3+\delta} +  \varepsilon^2  I^2_{N+2}  t^{-1} \|  \nabla^{\mathring{N}}  \Delta^{ [\frac{N}{2}] } \Ew\|_{L^2}. 
 }
 
In summary,  for the coupling terms $S_l$, we have
\begin{align*}
  \int_{M_t}  |S_{l}|\, \di \mu_g & \lesssim \varepsilon^3 \Lambda^3 t^{-3+2\delta}, \quad 0 \leq l \leq N-1, \\
  \int_{M_t}   |S_{N}|\, \di \mu_g & \lesssim \varepsilon^3 \Lambda^3 t^{-3+2\delta} + \varepsilon^2  I^2_{N+2}   t^{-1} E^{\frac{1}{2}}_{[N]} (\W). 
\end{align*}

From the energy identity, we derive,
\begin{equation}\label{eq-energy-inequality-Bianchi-leq-2}
 \dtau E_{N-1}(\W,t) + 2 E_{N-1}(\W, t)  \lesssim \varepsilon^3 \Lambda^3 t^{-3+2\delta},
\end{equation}
 for $0<\delta < \frac{1}{6}$. It follows that
\begin{equation*}
 \dt (t^2 E_{N-1}(\W,t))  \lesssim \varepsilon^3 \Lambda^3 t^{-2+2\delta},
\end{equation*}
and hence
\be
t^2 E_{N-1}(\W,t) \lesssim  \varepsilon^2 I_{N+1}^2 +  \varepsilon^3 \Lambda^3.  \label{eq-energy-estimate-Bianchi-l-leq-2}
\ee

When $l= N$, the estimate turns into
\als
{
 \dtau E_{[N]}(\W,t)  + 2 E_{[N]}(\W, t)  & \lesssim  \varepsilon I_{N+1}  E_{[N]}(\W, t)  \nnb \\
 & +  \varepsilon^3 I^3_{N+1}   t^{-2} +  \varepsilon^3 \Lambda^3 t^{-3+2\delta},
}
which is equivalent to 
\als{
 \dt (t^2 E_{[N]}(\W,t)) & \lesssim  \varepsilon I_{N+2}  t^{-1} (t^2 E_{[N]}(\W, t)) \\
 & \quad +   \varepsilon^3 I^3_{N+2}  t^{-1} +  \varepsilon^3 \Lambda^3 t^{-2+2\delta}.
}
An application of the Gr\"{o}nwall's inequality then yields
\be 
t^2 E_{[N]}(\W,t) \lesssim \left(  \varepsilon^2 I_{N+2}^2 +  \varepsilon^3 \Lambda^3 \right)t^{ \varepsilon C(I_{N+2})}.\label{eq-energy-estimate-Bianchi-l-3}
\ee
 Now we complete the proof by applying Corollary \ref{prop-elliptic-Delta-1}.
\end{proof}

\subsubsection{Estimates for the second fundamental form}\label{sec-ee-2ff}

In this section, we make use of the transport equations \eqref{eq-evolution-xi}--\eqref{eq-evolution-2} and constraint equations \eqref{E:div-curl-k-hat} to estimate the decay of $\|\eta\|_{H_{N+1}}$ and $\| \Sigma\|_{H_{N+1}}$. 
It is crucial to estimate the trace part $\eta$ a priori, since the transport equation of $\eta$ admits a structure of saving regularity. This makes it possible to view the constraint equations  \eqref{E:div-curl-k-hat} as an elliptic system for the traceless part $\Sigma$. We note that the estimate of $\|\eta\|_{H_{N+2}}$ is neither needed nor achieved at this stage.

\begin{proposition}\label{prop-energy-estimate-2nd}
%Let $I>0$ be constant depending only on the initial data, 
Under the bootstrap assumptions \eqref{BT-E-H}--\eqref{BT-KG}, and $N \geq 2$, we have
\begin{equation}\label{ee-xi}
 t^2 \|\eta\|^2_{H_{N+1}}  \lesssim \varepsilon^2  I_{N+2}^2 + \varepsilon^4 \Lambda^4 +  \varepsilon^2  I_{N+2}^2  t^{\varepsilon C(I_{N+2})}
 \end{equation}
 and
\begin{equation}\label{eq-energy-hat-k}
t^2 \| \Sigma\|^2_{H_{N+1}}  \lesssim  \left(  \varepsilon^2 I_{N+2}^2 +  \varepsilon^3 \Lambda^3 \right)  t^{ \delta}.
\end{equation}
\end{proposition}

\begin{proof}
The transport equation \eqref{eq-evolution-xi} includes terms such as $\eta^2$, $|\Sigma|^2$ and $\D \phi * \D \phi$  on the right hand side, which are of the same regularity as $\eta$ under the bootstrap assumptions \eqref{BT-E-H}--\eqref{BT-KG}.
This structure helps us to save regularities of $\eta$. In fact, 
\als{ 
\p_{\tau} \|\eta\|_{H_{N+1}}^2 +2 \|\eta\|_{H_{N+1}}^2 & \lesssim \left( \|\eta\|_{L^\infty} +  \|\Sigma\|_{L^\infty} \right) \|\eta\|_{H_{N+1}}^2 \\
&+\left(  \|\eta^2\|_{H_{N+1}}  +  \|\Sigma * \Sigma\|_{H_{N+1}}  +   \|\D \phi * \D \phi\|_{H_{N+1}} \right)  \|\eta\|_{H_k}.
}
Applying the bootstrap assumptions, we have
\als{ 
\p_{t} (t^2 \|\eta\|_{H_{N+1}}^2)  & \lesssim \varepsilon \Lambda t^{-2+ \delta}  \cdot t^2 \|\eta\|^2_{H_{N+1}}+  \left( \varepsilon^2 \Lambda^2 t^{-2+2 \delta}  +  \varepsilon^2  I^2_{N+2}  t^{-1}  \right)  \cdot t \|\eta\|_{H_{N+1}},
}
%That is,  for $0 \leq k \leq N+1$,
%\[\p_{t} (t \|\eta\|_{H_k})  \lesssim \varepsilon \Lambda t^{-2+ \delta}  \cdot t \|\eta\|_{H_k} + \varepsilon^2 \Lambda^2 t^{-2+2 \delta}  +  \varepsilon^2  I^2_{N+2}  t^{-1},\]
and \eqref{ee-xi} follows.

From the transport equation \eqref{eq-evolution-2}, and the fact \[ \R_{i \tau j \tau} \Sigma^{ij} = (\Ew_{ij} - \frac{1}{2} \nabla_i \phi \nabla_j \phi) \Sigma^{i j}, \] we obtain
\als{
 \dtau \|\Sigma\|_{L^2}^2 +  2 \|\Sigma\|_{L^2}^2 &\lesssim \left( \|\eta\|_{L^\infty} +  \|\Sigma\|_{L^\infty} \right) \|\Sigma\|_{L^2}^2 \\
 &+  \left(\|\Ew \|_{L^2} + \|\nabla \phi\|_{L^\infty} \|\nabla \phi\|_{L^2} \right) \|\Sigma\|_{L^2}.
}
From the bootstrap assumptions \eqref{BT-E-H}--\eqref{BT-KG} and  Propositions \ref{prop-ee-kg}--\ref{prop-energy-estimate-Bianchi},
\als{ 
\p_{t} (t^2 \|\Sigma\|_{L^2}^2)  & \lesssim \varepsilon  \Lambda  t^{-2+ \delta} \cdot t^2 \|\Sigma\|_{L^2}^2 + \left(\varepsilon  I_{2} +  \varepsilon^{\frac{3}{2}} \Lambda^{\frac{3}{2}} +   \varepsilon^2 I^2_{4}  t^{-1} \right)  t^{-1}  \cdot t \|\Sigma\|_{L^2}.
%& \lesssim \varepsilon  \Lambda  t^{-2+ \delta} \cdot t^2 \|\Sigma\|_{L^2}^2 + \left( \varepsilon  I_{2} +  \varepsilon^{\frac{3}{2}} \Lambda^{\frac{3}{2}} +   \varepsilon^2 I^2_{4}  t^{-1} \right) \cdot \varepsilon  \Lambda  t^{-1+\delta}.
}
As a consequence, for a fixed constant $0<\delta < \frac{1}{6}$,
\be\label{ee-Sigma}
t \|\Sigma\|_{L^2} \lesssim ( \varepsilon  I_{2} + \varepsilon^{\frac{3}{2}} \Lambda^{\frac{3}{2}} ) \ln t \lesssim   ( \varepsilon  I_{2} + \varepsilon^{\frac{3}{2}} \Lambda^{\frac{3}{2}} ) t^{\delta}. %  +  \varepsilon^2 I^2_{4}.
\ee

With $\|\eta\|_{H_{N+1}}$ bounded, the div-curl equations \eqref{E:div-curl-k-hat} can be viewed as an elliptic system for $\Sigma$,
\als
{
(\dive \Sigma)_i &= 2 \nabla_i \eta - \nabla_j \phi   \dtau \phi, \nnb \\
(\curl \Sigma)_{ij} &= - \Hw_{ij}  - \epsilon_{ij}{}^{\!m} \nabla_m \eta,  
}
where we use the fact $\bar R_{\tau j} = \breve R_{\tau j}  = \breve{T}_{\tau j} = \nabla_j \phi   \dtau \phi$. Thus,
\begin{equation*}
\int_{M} |\nabla \Sigma|^2 \, \di \mu_g = \int_{M} -3 R_{mn} \Sigma^{im} \Sigma_{i}^{\, n} + \frac{R}{2} |\Sigma |^2 + |\Hw_{ij}|^2 + 3 |\nabla \eta| + \frac{1}{2} | \nabla_j \phi   \dtau \phi|^2 \, \di \mu_g,
\end{equation*}
and hence the identity 
\begin{equation*}
\| \Sigma\|^2_{H_{k+1}} \lesssim \| \Sigma\|^2_{H_{k}} + \| \Hw\|^2_{H_{k}} + \| \eta\|^2_{H_{k+1}} +\|  \nabla \phi   \dtau \phi \|^2_{H_{k}}, \quad k \leq N,
\end{equation*}
due to the bootstrap assumptions \eqref{BT-E-H}--\eqref{BT-KG}.
Combining \eqref{ee-Sigma} and  Propositions \ref{prop-ee-kg} and \ref{prop-energy-estimate-Bianchi}, the estimate \eqref{eq-energy-hat-k} follows.
%\als
%{t^2 \| \Sigma\|^2_{H_{N+1}} & \lesssim ( \varepsilon^2 I^2 + \varepsilon^2 \Lambda I + \varepsilon^{\frac{5}{2}} \Lambda^{\frac{5}{2}} ) t^{2 \delta} + \varepsilon^2 I^2 + \varepsilon^3 \Lambda^3.
%}
\end{proof}

\begin{remark}\label{rk-ee-closed}
Based on Propositions \ref{prop-ee-kg}, \ref{prop-energy-estimate-Bianchi} and \ref{prop-energy-estimate-2nd}, a further application of equations \eqref{eq-evolution-1}--\eqref{eq-evolution-2} and \eqref{Gauss-Ricci-hat-k} implies
\begin{align}
t^2 \| \dtau g_{ij} \|^2_{H_{N+1}} & \lesssim \left(  \varepsilon^2 I_{N+2}^2 +  \varepsilon^3 \Lambda^3 \right)  t^{\delta}, \label{ee-tg} \\
\|g_{ij} - \gamma_{i j}\|^2_{H_{N+1}} & \lesssim \varepsilon^2 I_{N+2}^2 +  \varepsilon^3 \Lambda^3, \label{ee-g-lower} \\
 t^2 \|\dtau \Sigma\|^2_{H_{N}} & \lesssim \left(  \varepsilon^2 I_{N+2}^2 +  \varepsilon^3 \Lambda^3 \right)  t^{\delta}, \label{ee-pt-Sigma} \\
  t^2 \|\dtau \eta \|^2_{H_{N}} & \lesssim \varepsilon^2  I_{N+2}^2 + \varepsilon^4 \Lambda^4 + \varepsilon^2  I_{N+2}^2 t^{ \varepsilon C(I_{N+2})}, \label{ee-pt-eta} \\
t^2 \| R_{ij} + 2 g_{ij } \|^2_{H_{N}} & \lesssim  \left(  \varepsilon^2 I_{N+2}^2 +  \varepsilon^3 \Lambda^3 \right)  t^{\delta}.  \label{ee-Ric}
\end{align} 
\end{remark}

\subsection{Closure of the bootstrap argument}
Now the main energy estimates, Propositions \ref{prop-ee-kg}--\ref{prop-energy-estimate-2nd} (also \eqref{ee-g-lower}), have been established. %we achieve that $\Lambda^2 \lesssim  I_{N+2}^2  +   I_{2} \Lambda  +  \varepsilon \Lambda^3 + \varepsilon^{\frac{1}{2}} \Lambda^{\frac{5}{2}}$.
We can take $\Lambda > I_{N+2}$ large enough and $\varepsilon$ small enough (depending on $I_{N+2}$), so that the bootstrap assumptions \eqref{BT-E-H}--\eqref{BT-KG} hold with $\Lambda$ replaced by $\frac{1}{2} \Lambda$.  In fact, this is the main difference from the technics of  local existence theory \cite{Fournodavlos-Luk-Kasner}. Namely, our main energy estimates are conducted with the regularity $(\eta, \Sigma) \in H_{N+1} \times H_{N+1}$, instead of $(\eta, \Sigma) \in H_{N+2} \times H_{N+1}$ \cite{Fournodavlos-Luk-Kasner}, so that all of the quantities $\|\Ew \|_{H_{N}}$, $\|\Hw \|_{H_{N}}$, $\|\Sigma \|_{H_{N+1}}$, $\|\eta \|_{H_{N+1}}$ decay well, and hence we can close the main energy estimates.
%We thus close the bootstrap argument. 

Note that, the local existence theorem \cite{Fournodavlos-Luk-Kasner} requires a bound for the top order $\nabla_{I_{N+2}} \eta$, which is not covered in the previous estimates. The estimate for $\nabla_{I_{N+2}} \eta$ will be postponed in the next subsection. Once the bound for $\nabla_{I_{N+2}} \eta$ is established,  
%Combining with Proposition \ref{pro-improve-eta} and 
together with the local existence theorem  \ref{thm-local-existence}, we conclude Theorem \ref{Thm-main-ee-EKG},
%there is a constant $C(I_{N+2})$ depending on $I_{N+2}$ such that  
%\als{
%t^{1-\delta} ( \|\Ew \|_{H_{N}} +  \|\Hw \|_{H_{N}} & +  \|\eta \|_{H_{N+1}}   +  \|\Sigma \|_{H_{N+1}}   ) (t) \nnb \\ 
%&+ t^{\frac{1}{2}} \|\D \phi \|_{H_{N+1}}  (t)   \leq \varepsilon C(I_{N+2}).}
and the global existence for the massive Einstein-scalar field equations  in Theorem \ref{thm-global-existence} follows.

\subsection{Improvement of the regularity for $\eta$}\label{sec-improve-regularity}
In this section we follow the idea in \cite{Fournodavlos-Luk-Kasner} to retrieve a uniform bound for the top order $\nabla_{I_{N+2}} \eta$. Note that, due to the presence of the massive scalar field, we lost the decay estimate for $\nabla_{I_{N+2}} \eta$.  However, in the massless case, there is no loss in decay rate, see Proposition \ref{prop-improve-w}. % and $\|g_{ij} - \gamma_{i j}\|_{H_{N+2}}$.
\begin{proposition}\label{pro-improve-eta}
For $N \geq 2$, we have
\als{
 \|\nabla_{I_{N+2}} \eta\|_{L^2} & \lesssim ( \varepsilon I_{N+2} + \varepsilon^2 \Lambda^2 ) t^{-1} + \varepsilon^2 I_{N+2}^2.
%\|g_{ij} - \gamma_{i j} \|_{H_{N+2}} &  \lesssim  \varepsilon I_{N+2}  t^{ \varepsilon I_{N+2} }.
}
\end{proposition}

\begin{proof}
Applying $\nabla_{I_{N}} \Delta$ on \eqref{eq-evolution-xi} and commute it with $\p_\tau$, we have
\als{
 \p_{\tau}  \nabla_{I_{N}} \Delta \eta +   \nabla_{I_{N}} \Delta \eta &=  2 \eta  \nabla_{I_{N}} \Delta \eta +  \frac{2}{3}  \nabla_{I_{N}} \Delta \Sigma_{i j} \cdot \Sigma^{i j} + \frac{1}{3}   \nabla_{I_{N }} \Delta \R_{\tau \tau} (\phi) \nnb \\
& \quad + \sum_{ \substack{a+b\leq N+2\\ a, \, b \leq N+1} } ( \nabla_{I_a} \Sigma, \nabla_{I_a} \eta) * ( \nabla_{I_b} \Sigma, \nabla_{I_b} \eta).
}
We notice that the terms $\frac{2}{3}  \nabla_{I_{N}} \Delta \Sigma_{i j} \cdot \Sigma^{i j}$ and $\frac{1}{3}   \nabla_{I_{N }} \Delta \R_{\tau \tau} (\phi)$ on the right hand side are not bounded due to the restriction of regularity. In the following, we will treat these two terms by means of renormalization.

From \eqref{eq-ricci-phi}, we have
\als{
  & \nabla_{I_{N }} \Delta \R_{\tau \tau} (\phi) =    \nabla_{I_{N }} \Delta \left(\dtau \phi \dtau \phi - \frac{1}{2} m^2 t^2 \phi^2 \right) \\
  ={} &2   \nabla_{I_{N }} \Delta \dtau \phi \cdot \dtau \phi -  m^2 t^2  \nabla_{I_{N }} \Delta \phi \cdot  \phi + \sum_{\substack{a+b=N+2\\ a, \, b \leq N+1}}  \nabla_{I_{a}}  \D \phi * \nabla_{I_{b}} \D \phi.
}
By using the Klein--Gordon equation, the top order term can be rearranged as below
\als{
 &2   \nabla_{I_{N }} \Delta \dtau \phi \cdot \dtau \phi -  m^2 t^2  \nabla_{I_{N }} \Delta \phi \cdot  \phi \\
  = {}& \dtau ( 2  \nabla_{I_{N }} \Delta  \phi \cdot \dtau \phi)  -    \nabla_{I_{N }} \Delta  \phi   ( 2 \dtau^2 \phi + m^2 t^2 \phi) - 2 [\dtau, \nabla_{I_N} \Delta] \phi \cdot \dtau \phi \\
 = {}& \dtau ( 2  \nabla_{I_{N }} \Delta  \phi \cdot \dtau \phi)  -    \nabla_{I_{N }} \Delta  \phi   ( 2 \Delta \phi - 4 \dtau \phi - m^2 t^2 \phi - 6 \eta \dtau \phi ) \\
 & + \sum_{a+b  = N+1}   \left( \nabla_{I_a}  \Sigma, \, \nabla_{I_a}  \eta \right) * \nabla_{I_{b}} \nabla \phi \cdot \dtau \phi.
}  
%[\dtau, \nabla_{I_l}](\psi) &= \sum_{a+2 +b  = l}  \nabla_{I_a}  \left( \nabla \Sigma_{i j} + \nabla \eta \right) * \nabla_{I_{b}} \nabla \psi, \quad l \geq 2, \label{def-commuting-KN-l}
It follows that, for $N \geq 2$,
\al{trans-source}{
   \nabla_{I_{N }} \Delta \R_{\tau \tau} (\phi) = {}& \dtau ( 2  \nabla_{I_{N }} \Delta  \phi \cdot \dtau \phi)  +  m^2 t^2  \nabla_{I_{N }} \Delta  \phi \cdot  \phi  + f_0,
} 
with \[  \|f_0 \|_{L^2} \lesssim \varepsilon^2 I^2_{N+2} t^{-1} + \varepsilon^3 \Lambda^3 t^{-2+\delta}. \]
On the other hand, using the wave equation \eqref{E:wave-Sigma}, we obtain
\als{
& \nabla_{I_{N}} \Delta \Sigma_{i j} \cdot \Sigma^{i j} =  
\dtau \left( \dtau  \nabla_{I_{N}} \Sigma_{i j} \Sigma^{i j} + 2   \nabla_{I_{N}} \Sigma_{i j} \Sigma^{i j} \right) \\
&\quad + 3 \nabla_{I_{N}} \nabla_i \nabla_j \eta \Sigma^{i j}  + \Sigma^{ij}   \nabla_{I_{N}} \left(\dtau \R_{ij} (\phi) - \nabla_i (\dtau \phi \nabla_j \phi) - \nabla_j (\dtau \phi \nabla_i \phi) \right) \\
&\quad + f_1, 
}
with \[  \|f_1\|_{L^2} \lesssim (\varepsilon^2 I^2_{N+2}+\varepsilon^3 \Lambda^3) t^{-2+\delta}. \] 
\begin{comment}
\als
{& \dtau^2 \Sigma_{ij} - \Delta \Sigma_{ij} + 2 \dtau \Sigma_{ij}  \\
={}& 3 \nabla_i \nabla_j \eta - \Delta \eta g_{i j} - \dtau \R_{ij} (\phi) + \nabla_i (\dtau \phi \nabla_i \phi) + \nabla_j (\dtau \phi \nabla_i \phi) \\
& + (\dtau \Sigma, \dtau \eta) * (\Sigma, \eta) +  (\Sigma, \eta) * (\Sigma, \eta)  \\
& +  (\Sigma, \eta) * (\Sigma, \eta) * (\Sigma, \eta)   \\
& +  (\Sigma, \eta) * (\Sigma, \eta) * (\R_{ m n}, \R_{\tau \tau})  \pm \R_{m n} \pm \R_{\tau \tau}  .
}
we have
\als{
& \p_{\tau} \nabla_{I_k} \Delta \eta +  \nabla_{I_k} \Delta \eta - \frac{2}{3}\dtau (  \nabla_{I_{N }} \Delta  \phi \cdot \dtau \phi) \\
&= \frac{2}{3} \left( \lie^2_{\dtau} \nabla_{I_k} \Sigma_{i j} + 2 \lie_{\dtau} \nabla_{I_k} \Sigma_{i j} \right)   \cdot \Sigma^{i j}   - \frac{2}{3} \Sigma^{ij}  \nabla_{I_k} \nabla_i \nabla_j \eta \\
& - \frac{2}{3}\Sigma^{ij}  \nabla_{I_k} (\dtau \R_{ij} - \nabla_i (T \phi \nabla_j \phi) - \nabla_j (T \phi \nabla_i \phi)) \nnb \\
& \quad \C{+  2  m^2 t^2  \nabla_{I_{N }} \Delta  \phi \cdot  \phi } + 2 \nabla_{I_k} \Delta \eta \cdot \eta  + \sum_{  a+b\leq k  } ( \nabla_{I_a}\dtau \Sigma, \nabla_{I_a} \dtau \eta) * ( \nabla_{I_b} \Sigma, \nabla_{I_b} \eta) \nnb \\
& \quad  + \sum_{ \substack{a+b\leq k+2\\ a, \, b \leq k+1} } ( \nabla_{I_a} \Sigma, \nabla_{I_a} \eta) * ( \nabla_{I_b} \Sigma, \nabla_{I_b} \eta)  + \sum_{\substack{a+b\leq k+2\\ a, \, b \leq k+1}}  \nabla_{I_{a}}  \D \phi \cdot \nabla_{I_{b}} \D \phi + l.o.t.,  \\
}
\end{comment}
According to the equation \eqref{eq-ricci-phi},  
\als{
 & \dtau \R_{ij} (\phi) - \nabla_i (\dtau \phi \nabla_j \phi) - \nabla_j (\dtau \phi \nabla_i \phi) \\
  =& \dtau  \left(  \frac{m^2 t^2}{2} \phi^2  g_{i j} \right) - 2 \dtau \phi \nabla_i \nabla_j \phi.
}
Applying the fact that $\Sigma$ is trace free, yields
\als{ 
& \Sigma^{ij}   \nabla_{I_{N}} (\dtau \R_{ij} (\phi) - \nabla_i (\dtau \phi \nabla_j \phi) - \nabla_j (\dtau \phi \nabla_i \phi)) \\
&=2 \Sigma^{ij}   \nabla_{I_{N}}  (\dtau \phi \nabla_i \nabla_j \phi) + \frac{1}{2} \Sigma^{ij}   \nabla_{I_{N}} \left( m^2 t^2 \phi^2 \dtau g_{i j} \right) \\
&= \sum_{a+b = N+1} \Sigma * \nabla_{I_{a}} \D \phi * \nabla_{I_{b}} \D \phi + \sum_{a+b+c= N} \Sigma * \nabla_{I_{a}} \D \phi * \nabla_{I_{b}} \D \phi * \nabla_{I_{c}} \Sigma. 
}
That is,  
\al{trans-Sigma}{
 \nabla_{I_{N}} \Delta \Sigma_{i j} \cdot \Sigma^{i j} = {}&  
\dtau ( \dtau  \nabla_{I_{N}} \Sigma_{i j} \Sigma^{i j} + 2   \nabla_{I_{N}} \Sigma_{i j} \Sigma^{i j} ) + 3 \nabla_{I_{N}} \nabla_i \nabla_j \eta \Sigma^{i j} \nnb \\
&+   f_2, 
}
with \[ \|f_2\|_{L^2} \lesssim (\varepsilon^2 I^2_{N+2}+\varepsilon^3 \Lambda^3) t^{-2+\delta}. \]

With the help of equations \eqref{E:trans-source} and \eqref{E:trans-Sigma}, we derive
\als{
& \p_{\tau}  \left(  \nabla_{I_N} \Delta \eta - \frac{2}{3} \left( \dtau  \nabla_{I_{N}} \Sigma_{i j} \Sigma^{i j} + 2  \nabla_{I_N} \Sigma_{i j} \Sigma^{i j}+ \nabla_{I_{N }} \Delta  \phi \cdot \dtau \phi  \right)  \right)  +  \nabla_{I_N} \Delta \eta   \\
 ={}& 2 \eta \nabla_{I_N} \Delta \eta + 2 \Sigma \cdot \nabla_{I_{N+2}} \eta +  2  m^2 t^2  \nabla_{I_{N }} \Delta  \phi \cdot  \phi   +   f_3, 
}
with \[ \|f_3\|_{L^2} \lesssim \varepsilon^2 I^2_{N+2} t^{-1} + \varepsilon^3 \Lambda^3 t^{-2+\delta}. \]
Let
\begin{equation}\label{def-ti-eta}
  \tilde \eta_{N+2} =\nabla_{I_N} \Delta \eta - \frac{2}{3} \left( \dtau  \nabla_{I_{N}} \Sigma_{i j} \Sigma^{i j} + 2  \nabla_{I_N} \Sigma_{i j} \Sigma^{i j}+ \nabla_{I_{N }} \Delta  \phi \cdot \dtau \phi  \right).
  \end{equation}
We obtain
\als{
 \p_{\tau}   \tilde \eta_{N+2} +  \tilde \eta_{N+2}  ={}& 2 \nabla_{I_N} \Delta \eta \cdot \eta  + 2 \Sigma \nabla_{I_{N+2}} \eta +  2  m^2 t^2  \nabla_{I_{N }} \Delta  \phi \cdot  \phi  +   f_4,
}
with \[ \|f_4\|_{L^2} \lesssim \varepsilon^2 I^2_{N+2} t^{-1} + \varepsilon^2 \Lambda^2 t^{-2+2\delta}. \]
As a consequence, for $N \geq 2$, 
\als{
\p_t (t^2  \|\tilde \eta_{N+2}\|^2_{L^2} ) & \lesssim t^2  \|\tilde \eta_{N+2}\|^2_{L^2} \cdot t^{-1} (\|\Sigma\|_{L^\infty} + \| \eta \|_{L^\infty}) \\
& + t \|\tilde \eta_{N+2}\|_{L^2} \|\nabla_{I_{N+2}} \eta\|_{L^2} \cdot  (\|\Sigma\|_{L^\infty} + \| \eta \|_{L^\infty}) \\
%& + t \tE^{\frac{1}{2}}_{[N+2]} (\eta) \cdot  ( \|\Sigma\|^2_{H_{N+1}} + \|\dtau \Sigma\|^2_{H_{N}}  + \|\eta\|^2_{H_{N+1}} + \|\dtau \eta \|^2_{H_{N}}  + \|\D \phi \|^2_{H_{N+1}} ) \\
& + t  \|\tilde \eta_{N+2}\|_{L^2} \cdot  ( \varepsilon^2 I^2_{N+2} t^{-1} + \varepsilon^2 \Lambda^2 t^{-2+2\delta} + t  \|\nabla_{I_{N }} \Delta  \phi \cdot mt \phi \|_{L^{2}} ).
}
By the known estimates for $\|\Sigma\|_{L^\infty}$, $ \| \eta \|_{L^\infty}$ \eqref{E:pre-estimates}, and \[t  \|\nabla_{I_{N }} \Delta  \phi \cdot mt \phi \|_{L^{2}} \lesssim t \|\nabla_{I_{N }} \Delta  \phi \|_{L^{2}} \|mt \phi \|_{L^{\infty}} \lesssim \varepsilon^2 I^2_{N+2}, \] it follows that
%\als{
%\p_t (t^2  \|\tilde \eta_{N+2}\|^2_{L^2}  ) & \lesssim t^2  \|\tilde \eta_{N+2}\|^2_{L^2}  \cdot  \varepsilon I_{N+2} t^{-2+\delta} \\
%& + t \|\tilde \eta_{N+2}\|_{L^2} \|\nabla_{I_{N+2}} \eta\|_{L^2} \cdot  \varepsilon I_{N+2} t^{-1+\delta} \\
%& + t  \|\tilde \eta_{N+2}\|_{L^2} \cdot  ( \varepsilon^2 I^2_{N+2} t^{-1} + \varepsilon^2 I^2_{N+2} ) 
%}
%and then 
\als{
\p_t (t  \|\tilde \eta_{N+2}\|_{L^2}  ) & \lesssim ( \|\tilde \eta_{N+2}\|_{L^2} +  \|\nabla_{I_{N+2}} \eta\|_{L^2} )   \varepsilon \Lambda t^{-1+\delta} \\
%& +  \|\nabla_{I_{N+2}} \eta\|_{L^2} \cdot  \varepsilon I_{N+2} t^{-1+\delta} \\
& +  \varepsilon^2 I^2_{N+2} t^{-1} + \varepsilon^2 \Lambda^2 t^{-2+2\delta} + \varepsilon^2 I^2_{N+2}. 
}
By Corollary \ref{prop-elliptic-Delta-1},
\als{
\|\nabla_{I_{N+2}} \eta\|_{L^2} \lesssim {}&  \|\nabla_{I_{N}} \Delta \eta\|_{L^2} + \| \eta\|_{H_{N+1} } \\
\lesssim & \|\tilde \eta_{N+2}\|_{L^2} +  \varepsilon^2 I^2_{N+2} t^{-1} +  \varepsilon^2 \Lambda^2 t^{-2+2\delta},
}
 we achieve the energy inequality,
\al{energy-inequ-teta}{
t  \|\tilde \eta_{N+2}\|_{L^2}  \lesssim {} & \varepsilon I_{N+2} + \varepsilon^2 \Lambda^2 + \int_{t_0}^t \varepsilon^2 I^2_{N+2} \, \di t   \nnb \\
&+ \int_{t_0}^t \varepsilon \Lambda t^{-2+\delta} \cdot t\|\tilde \eta_{I_{N+2}}\|_{L^2} \, \di t.
}
It then follows from the Gr\"{o}nwall's inequality that
\[%\label{ee-N+2-xi}
t\|\tilde \eta_{I_{N+2}} \|_{L^2} \lesssim \varepsilon I_{N+2} + \varepsilon^2 \Lambda^2 +  \varepsilon^2 I^2_{N+2} t.
\]
\begin{comment}
Since $\|g_{ij} - \gamma_{i j}\|_{H_{N+1}}  \lesssim \varepsilon I_{N+2}$, and noting that $\nabla [\gamma] g = \nabla \gamma * g$, hence there is, 
\als{
\sum_{1\leq k \leq N+1}\| \nabla_{I_k} \gamma_{i j}\|_{L^2}  & \lesssim  \varepsilon I_{N+2}, \\
\sum_{1 \leq k\leq N+1} \| \nabla_{I_k} [\gamma] g_{i j}\|_{L^2(M, \gamma)} & \lesssim \varepsilon I_{N+2}.
}
%and hence $\sum_{1 \leq k\leq N+1} \| \nabla_{I_k} [\gamma] g_{i j}\|_{L^2}  \lesssim \varepsilon I_{N+2}$, 
Then
 we have 
 \als{
 \|\nabla_{I_{N+2}} \eta\|_{L^2} & \lesssim \| \eta\|_{H_{N+2}(M, \gamma)}  \\
 & \lesssim  \| \eta\|_{H_{N+1}(M, \gamma)} + \|\nabla_{I_{N}} [\gamma] \Delta_\gamma  \eta\|_{L^2(M, \gamma)} \\
 & \lesssim \| \eta\|_{H_{N+1} } +  \|\nabla_{I_{N}} \Delta \eta\|_{L^2}.
 }
where the second inequality follows by Proposition \ref{prop-elliptic-Delta} and the density theorem on $(M, \gamma)$. 
\end{comment}
Since \[ \|\nabla_{I_{N+2}} \eta\|_{L^2} \lesssim  \|\tilde \eta_{N+2}\|_{L^2} +  \varepsilon^2 I^2_{N+2} t^{-1} + \varepsilon^2 \Lambda^2 t^{-2+2\delta}, \] it concludes our proof.
\end{proof}

\begin{remark}
Based on the equation \eqref{eq-evolution-1} and Proposition \ref{pro-improve-eta}, we know that $\|\nabla_{I_{N+2}} ( g-\gamma )\|_{L^2}$ has $\varepsilon \ln t$ growth. Combined this fact with the non-decaying estimate for $\|\nabla_{I_{N+2}} \eta\|_{L^2}$  makes it difficult to close the energy argument using the reduced system (or rather the second order equation for $\Sigma$ \eqref{E:wave-Sigma}).
\end{remark}

\section{Global existence for massless Einstein-scalar field}\label{sec-EW}
In this section, we consider the massless Einstein-scalar field equations, i.e. $m=0$.
Similar to the massive case, we will prove the following propagation estimates.
\begin{theorem}[Main estimates for massless Einstein scalar field]\label{ME-w}
Fix two independent constants $\delta, \, \sigma \in \left( 0, \,  \frac{1}{6} \right)$ and an integer $N \geq 2$.
For $\varepsilon$ small enough, and all initial data of  massless Einstein-scalar field equations and $I_{N+2} \in \mathbb{R}^{+}$ satisfy
\als{
&  \| \dtau \phi \|_{H_{N+1}}  (t_0)  +  \| \phi \|_{H_{N+2}}  (t_0) +  \|g_{ij} - \gamma_{ij} \|_{H_{N+1}}  (t_0)   \\
& \quad + \|\eta \|_{H_{N+2}}  (t_0) +  \|\Sigma \|_{H_{N+1}}  (t_0)  + \|\Ew \|_{H_{N}} (t_0) +  \|\Hw \|_{H_{N}}  (t_0)   \leq \varepsilon I_{N+2},
}
%and \[\frac{7}{8} \gamma_{ij} \leq g_{ij} (t_0) \leq \frac{9}{8} \gamma_{ij} \quad \text{ as bilinear forms}, \] 
there is a constant $C(I_{N+2})$ depending only on $I_{N+2}$ (in particular, not on $\varepsilon$ and $t_0$), such that 
\al{energy-estimate}
{
&  t^{1-\sqrt{\sigma}}( \| \phi \|_{H_{N+2}} + \| \dtau \phi \|_{H_{N+1}})  (t) + t^{1-\delta/2} \|\Sigma \|_{H_{N+1}} (t)  \nnb \\
&+ t  ( \|\Ew \|_{H_{N}} +  \|\Hw \|_{H_{N}} + \|\eta \|_{H_{N+2}}) (t) +  \|g_{ij} - \gamma_{ij} \|_{H_{N+1}}(t)  \leq \varepsilon C(I_{N+2}),
}
and  
\als{
  t^{1-\delta/2} \| R_{ij} + 2 g_{ij } \|_{H_{N}} (t) &\leq \varepsilon C(I_{N+2}),
   % t^{1-\delta/2} \| R_{imjn} + g_{ij } g_{m n} - g_{in} g_{jm} \|_{H_{N}} &  \leq \varepsilon C(I_{N+2}).
  }
  hold for all $t\geq t_0$.
%\[ t^{1-\delta} (  \|\dtau \eta \|_{H_{N+1}}   +  \|\dtau \Sigma \|_{H_{N}}   ) (t) \leq \varepsilon C(I_{N+2}). \]
%The subindex $N+2$ denotes the number of derivatives used in the norms.
\end{theorem}

\subsection{Bootstrap assumptions}
Let $\delta, \, \sigma$ and $N$ be constants satisfying the assumptions in Theorem \ref{ME-w}, and $\Lambda$ be a large constant to be determined later. We assume 
%We start with the following weak assumptions:
\begin{subequations}
\begin{align}
t ( \|\Ew \|_{H_{N}} +  \|\Hw \|_{H_{N}} ) & \leq \varepsilon \Lambda, \label{wBT-E-H} \\
t \|\eta \|_{H_{N+2}} + t^{1-\delta} \|\Sigma \|_{H_{N+1}} & \leq  \varepsilon  \Lambda,  \label{wBT-sigma-eta} \\
\|\Ew \|_{C^0} + \|\eta \|_{C^0} + \| \Sigma \|_{C^0} + \| g_{ij} - \gamma_{i j} \|_{C^1}  + \|D \phi \|_{C^0} & \leq \varepsilon \Lambda, \label{wBT-g} \\
  t^{1-\sqrt{\sigma}} ( \|\dtau \phi \|_{H_{N+1}} +\| \phi \|_{H_{N+2}} )  & \leq \varepsilon \Lambda. \label{wBT} 
 % \|g_{ij} - \gamma_{i j}\|_{H_{N+2}} & \leq \varepsilon \Lambda,\label{bt-g-gamma-w}
\end{align}
\end{subequations}

As before, the bootstrap assumptions \eqref{wBT-E-H}--\eqref{wBT} imply that $g_{ij}$ and $\gamma_{ij}$ are equivalent as bilinear forms (and $\| g- \gamma \|_{H_{N+1}}  \lesssim  \varepsilon  \Lambda$),
%\begin{align}\label{bt-g} 
% \frac{3}{4} \gamma_{ij} < g_{ij} < \frac{5}{4} \gamma_{ij}
%\end{align}
%as bilinear forms, as $\varepsilon$ is small enough. 
and the spatial manifold $(M^3, \, g)$ has infinite injectivity radius with its Ricci tensor bounded from below. Then the density corollary \ref{coro-density} and the Sobolev inequalities on $(M^3,\, g)$ hold automatically. Namely, we have
\al{wpre-estimates}{
& t ( \|\Ew  \|_{C^{N-2}}+  \|\Hw  \|_{C^{N-2}}+ \|\Ew \|_{H^4_{N-1}} + \|\Hw \|_{H^4_{N-1}} ) \lesssim \varepsilon \Lambda, \nnb \\
& t  ( \|\eta  \|_{C^{N}} + \|\eta \|_{H^4_{N+1}} ) + t^{1-\delta} (\|\Sigma  \|_{C^{N-1}}+ \|\Sigma \|_{H^4_{N}} )  \lesssim \varepsilon \Lambda, \nnb \\
&  t^{1-{\sqrt{\sigma}}}(\| \dtau \phi \|_{C^{N-1}} + \| \phi \|_{C^{N}} + \| \dtau \phi \|_{H^4_{N} } + \| \phi \|_{H^4_{N+1} }) \lesssim \varepsilon \Lambda.
}
In addition, it follows from \eqref{wBT-E-H}--\eqref{wBT},  \eqref{E:Gauss-Riem-hat-k}, \eqref{E:Riem-Weyl}, \eqref{eq-ricci-phi} and \eqref{E:wpre-estimates} that $\|R_{imjn}\|_{C^{N-2}}$ and $\|\nabla R_{imjn} \|_{H_{N-1}}$ are bounded for $N \geq 2$. Hence, Corollary \ref{prop-elliptic-Delta-1} holds true as well.

It is well known that $$\text{spec} - \Delta_\gamma \subset [1, + \infty)$$ for the canonical metric on ${\bf H}^3$ \cite[Chapter I, Theorem 5]{Chavel}. 
From Rayleigh's theorem, it implies \[\int_M |\nabla_\gamma \psi|_\gamma^2 \di \mu_\gamma \geq \int_M |\psi|^2 \di \mu_\gamma, \]  for any $\psi \in H_1(M, \gamma)$, 
Since $\|g_{ij} - \gamma_{i j}\|_{L^\infty} \leq \varepsilon \Lambda$ and $\varepsilon$ is sufficiently small,
we can find a constant
 \be\label{def-lambda-0}
\lambda_0 = 1-\sigma, \quad 0<\sigma<\frac{1}{6},
\ee
such that for any $\psi \in H_1(M)$, 
 \be\label{eigen-g}
 \int_M |\nabla \psi|^2\, \di \mu_g > \lambda_0 \int_M  |\psi|^2\, \di \mu_g.
 \ee

With the help of \eqref{eigen-g}, we are able to derive $t^{-1+\sqrt{\sigma}}$ decay estimate for the massless scalar field.

\subsection{Energy estimates}
We follow \cite{A-M-11-cmc} to define a modified wave type energy for the massless scalar field,
\begin{align*}
E^w_{[k+1]} (\phi, t) & =  \int_{M_t} \left( |\nabla^{\mathring{k}} \Delta^{ [\frac{k}{2}] } \dtau \phi|^2 +  |\nabla^{ k^\prime} \Delta^{ [\frac{k+1}{2}] }  \phi|^2  \right) \di \mu_g \\
& + \int_{M_t} 2 c_E \nabla^{\mathring{k}} \Delta^{ [\frac{k}{2}] } \dtau \phi \cdot \nabla^{\mathring{k}} \Delta^{ [\frac{k}{2}] } \phi \, \di \mu_g.
\end{align*}
where $$c_E = \lambda_0$$ and the notations $\mathring{k}$ and $k^\prime$ are defined in \eqref{def-k-prime}. The energy norm up to $l$ order,
is defined as in \eqref{def-energy-sum-inhomo}.
By the choice of $c_E$, we have (refer to Remark \ref{rk-positive-denergy-w} below)
\[ \|\phi\|_{H_{k+1}} + \|\dtau \phi\|_{H_k} \lesssim E^w_{k+1} (\phi, t), \quad k \in \mathbb{Z}_{\geq 0}.  \]
\begin{remark}\label{rk-positive-denergy-w}
$E^w_{[k+1]}$ is positive definite. For instance, when $k=0$, 
\als{
E^w_{[1]} (\phi, t)  =&  \int_{M_t} \left( | \dtau \phi|^2 +  |\nabla  \phi|^2 + 2 \lambda_0   \dtau \phi \cdot  \phi \right) \di \mu_g \\
\geq &   \int_{M_t} \left( | \dtau \phi|^2 + \lambda_0 |\phi|^2 + 2 \lambda_0   \dtau \phi \cdot  \phi \right) \di \mu_g, 
}
and 
when $k=1$, 
\als{
E^w_{[2]} (\phi, t)  =&  \int_{M_t} \left( |\nabla \dtau \phi|^2 +  |\Delta  \phi|^2 + 2 \lambda_0  \nabla \dtau \phi \cdot \nabla \phi \right) \di \mu_g \\
\geq &   \int_{M_t} \left( \lambda_0 | \dtau \phi|^2 +  |\Delta\phi|^2 - 2 \lambda_0   \dtau \phi \cdot \Delta \phi \right) \di \mu_g.
}
Note that, both of the matrices 
$\begin{pmatrix}
		1 & \lambda_0 \\
		\lambda_0 & \lambda_0
	\end{pmatrix}
$, 
$\begin{pmatrix}
		\lambda_0 & -\lambda_0 \\
		-\lambda_0 & 1
	\end{pmatrix}
$ are positive definite when $0< \lambda_0 < 1$. The rest cases follow similarly.
\end{remark}

\subsubsection{Energy estimates for the massless scalar equation}\label{sec-ee-wave-1+3}

The high order equation for the massless scalar field reads
\al{eq-w-N-l-1+3-general-simply}
 { &\lie_{\dtau} \nabla^{\mathring{l}} \Delta^{ [\frac{l}{2}] } \dtau \phi + 2 \nabla^{\mathring{l}} \Delta^{ [\frac{l}{2}] }  \dtau \phi -  \nabla^{\mathring{l}} \Delta^{ [\frac{l}{2}] +1 }  \phi   \nnb \\
= {} & \nabla^{\mathring{l}} \Delta^{ [\frac{l}{2}] } ( \eta * \dtau \phi) + \sum_{a+1 +b  = l}   \left(   \nabla_{I_a} \Sigma, \,  \nabla_{I_a} \eta \right) * \nabla_{I_{b}} \nabla \dtau \phi.
   }

\begin{proposition}\label{prop-ee-wave}
Under the bootstrap assumptions \eqref{wBT-E-H}--\eqref{wBT}, we have the decay estimate for the massless scalar field, %fix some $0<\sigma<\frac{1}{4}$,
\[ \|\dtau \phi\|_{H_{N+1}} + \|\phi\|_{H_{N+2}} \lesssim \varepsilon I_{N+2} t^{-1+\sqrt{\sigma}}, \quad N \geq 2. \]
\end{proposition}

\begin{proof}
\begin{comment}
Multiply $2 \nabla^{\mathring{l}} \Delta^{ [\frac{l}{2}] } \dtau \phi$ on \eqref{E:eq-kg-N-l-1+3-general-simply}, and using the commuting lemma \ref{lemma-commuting-application}, we have
\als{
& \dtau |\nabla^{\mathring{l}} \Delta^{ [\frac{l}{2}] } \dtau \phi|^2 +  \dtau (\nabla^{\dot{l}} \Delta^{ [\frac{l}{2}] +1 }  \phi  \nabla^{\dot{l}} \Delta^{ [\frac{l}{2}] +1}  \phi)   \nnb \\
& \C{- 2 \nabla^{\mathring{l}} ( \Delta^{ [\frac{l}{2}] +1 }  \phi  \nabla^{\mathring{l}} \Delta^{ [\frac{l}{2}] } \dtau \phi) } + 4 |\nabla^{\mathring{l}} \Delta^{ [\frac{l}{2}] }  \dtau \phi|^2  \nnb \\
  = {} &  \sum_{a +b  = l}  \left(   \nabla_{I_a} \Sigma +   \nabla_{I_a} \eta \right) * \nabla_{I_{b}} D \phi \nabla_{I_{l}} D \phi.
}
And
\als{
& \dtau  (\nabla^{\mathring{l}} \Delta^{ [\frac{l}{2}]  }  \phi  \nabla^{\mathring{l}} \Delta^{ [\frac{l}{2}] } \dtau \phi)   \\
  = {} &  |\nabla^{\mathring{l}} \Delta^{ [\frac{l}{2}] }  \dtau \phi|^2 -2 \nabla^{\mathring{l}} \Delta^{ [\frac{l}{2}] } \dtau  \phi \nabla^{\mathring{l}} \Delta^{ [\frac{l}{2}] } \phi  + \nabla^{\mathring{l}} \Delta^{ [\frac{l}{2}] +1 }  \phi \nabla^{\mathring{l}} \Delta^{ [\frac{l}{2}]  }  \phi \nnb\\
  &+ \sum_{a +b+1  = l}  \left(   \nabla_{I_a} \Sigma +   \nabla_{I_a} \eta \right) * \nabla_{I_{b}} \nabla \phi  \nabla^{\mathring{l}} \Delta^{ [\frac{l}{2}] } \dtau  \phi + f_l \nabla^{\mathring{l}} \Delta^{ [\frac{l}{2}] } \phi.
}
\end{comment}
As in \cite{A-M-11-cmc}, we define
\als{
\alpha_+=1-\sqrt{1-\lambda_0}=1-\sqrt{\sigma}.
}
Note that, $\alpha_+ < c_E$.
Now from the energy identity,
\als{
&\dtau E^w_{[l+1]} (\phi, t)   \nnb \\
  = &  \int_{M_t} 2(c_E -2)  |\nabla^{\mathring{l}} \Delta^{ [\frac{l}{2}] }  \dtau \phi|^2\, \di \mu_g- \int_{M_t} 4 c_E \nabla^{\mathring{l}} \Delta^{ [\frac{l}{2}] } \dtau  \phi \nabla^{\mathring{l}} \Delta^{ [\frac{l}{2}] } \phi \, \di \mu_g \\
  &- \int_{M_t} 2 c_E  |\nabla^{ l^\prime} \Delta^{ [\frac{l+1}{2}] }  \phi|^2 \,\di \mu_g +  N_l,
}
where
\als{
  N_l = &  \int_{M_t}   \sum_{a+1 +b  = l}  \left(   \nabla_{I_a} \Sigma +   \nabla_{I_a} \eta \right) * \nabla_{I_{b}}  \nabla \phi \nabla^{\mathring{l}} \Delta^{ [\frac{l}{2}] }  \dtau \phi \, \di \mu_g \nnb \\
  & +   \int_{M_t}  \nabla^{I_l} \left( (\eta, \, \Sigma_{ij}) * D \phi \right) * \left( \nabla^{\mathring{l}} \Delta^{ [\frac{l}{2}] }  D \phi + \nabla^{\mathring{l}} \Delta^{ [\frac{l}{2}] }  \phi \right)  \di \mu_g,
}
we deduce,
\als{
 \dtau E^w_{[l+1]} (\phi, t)  = & -2 \alpha_+ E^w_{[l+1]} (\phi, t)  +  N_l \\
  &+ 2 \int_{M_t} (c_E + \alpha_+ -2)  |\nabla^{\mathring{l}} \Delta^{ [\frac{l}{2}] }  \dtau \phi|^2\, \di \mu_g \\
  &+ 2 \int_{M_t} 2 c_E ( \alpha_+ - 1) \nabla^{\mathring{l}} \Delta^{ [\frac{l}{2}] } \dtau  \phi \nabla^{\mathring{l}} \Delta^{ [\frac{l}{2}] } \phi \, \di \mu_g \\
  & + 2 \int_{M_t} (\alpha_+ - c_E ) |\nabla^{ l^\prime} \Delta^{ [\frac{l+1}{2}] }  \phi|^2\, \di \mu_g \\
  \leq &  -2 \alpha_+ E^w_{[l+1]} (\phi, t)  +  |N_l|,
}
where the sum of the last three integral terms on the right hand side is negative \cite{A-M-11-cmc}, due to our choices of $\alpha_+$ and $c_E$.  
Consequently, we achieve 
\[\p_t ( t^{2\alpha_+} E^w_{N+2} (\phi, t) ) \lesssim \varepsilon \Lambda t^{-2 + \delta} \cdot  t^{2\alpha_+} E^w_{N+2}  (\phi, t), \]
which leads to \[  t^{2\alpha_+} E^w_{N+2} (\phi, t) \lesssim \varepsilon^2 I^2_{N+2}. \] 
In view of Corollary \ref{prop-elliptic-Delta-1}, it holds that
\be\label{energy-B-w} 
\| \dtau \phi\|^2_{H_{N+1}} + \|\phi\|^2_{H_{N+2}}  \lesssim \varepsilon^2 I^2_{N+2} t^{-2\alpha_+}.
\ee
\end{proof}

\subsubsection{Energy estimates for the geometry}\label{sec-ee-metric}
The energies for $\Ew$ and $\Hw$ are defined as in \eqref{def-energy-l-homo-Weyl}--\eqref{def-energy-sum-inhomo} and their estimates are standard.
\begin{proposition}\label{prop-energy-estimate-Bianchi-w}
Under the bootstrap assumptions \eqref{wBT-E-H}--\eqref{wBT}, we have
\als{
t^2  ( \|\Ew\|^2_{H_{N}} + \|\Hw\|^2_{H_{N}})  &\lesssim  \varepsilon^2 I_{N+2}^2 +  \varepsilon^3 \Lambda^3, \quad N \geq 2. 
}
\end{proposition} 
\begin{proof}
Based on an energy identity similar to the one in the proof of Proposition \ref{prop-energy-estimate-Bianchi}, we obtain
\als{
&\dtau  E_{[l+1]} (\W, t)   + 2 E_{[l+1]} (\W, t)    \nnb \\
 & =   \int_{M_t}  \sum_{k \leq l} \nabla_{ I_{k}} \left( \eta  * \W + \Sigma  * \W  \right)  *   \nabla^{\mathring{l}} \Delta^{ [\frac{l}{2}] }\W \, \di \mu_g \nnb \\
  & +   \int_{M_t} \left(  \nabla^{I_l} \left(  \nabla D \phi *  D \phi \right)  +  \nabla^{I_l} \left( ( \Sigma, \eta) * D \phi *  D \phi \right)  \right) *\nabla^{\mathring{l}} \Delta^{ [\frac{l}{2}] }\W  \,  \di \mu_g \\
  & +   \int_{M_t}  \sum_{k \leq l} \nabla_{ I_{k-1}} \left( \Ew * \W + (\Sigma, \eta)  * (\Sigma, \eta) + D \phi * D \phi \right) * \nabla^{\mathring{k}} \Delta^{ [\frac{k}{2}] } \W \, \di \mu_g,
}
and hence, for $0<\delta < \frac{1}{6}$,
\begin{equation}\label{eq-energy-inequality-Bianchi-w-low}
 \dtau E_{N}(\W,t) + 2 E_{N}(\W, t)  \lesssim \varepsilon^3 \Lambda^3 ( t^{-3+2\delta} +  t^{-3+2\sqrt{\sigma}} ),
\end{equation}
which yields the bound.
%\begin{equation*}
% \dt (t^2 E_{N}(\W,t))  \lesssim \varepsilon^3 \Lambda^3 ( t^{-2+\delta} +  t^{-2+2\sqrt{\sigma}} ).
%\end{equation*}
\end{proof}

In analogy to Proposition \ref{prop-energy-estimate-2nd}, we use the transport equations for $\eta$ to obtain the decay estimate for $\|\eta\|^2_{H_{N+1}}$, and it enables us to carry out the elliptic estimates for $\Sigma$. 
\begin{proposition}\label{prop-energy-estimate-2nd-w}
Under the bootstrap assumptions \eqref{wBT-E-H}--\eqref{wBT}, there are, for $N \geq 2$,
\begin{align*}
t^2 \|\eta\|^2_{H_{N+1}}  & \lesssim \varepsilon^2  I_{N+2}^2 + \varepsilon^4 \Lambda^4,  \\
t^2 \| \Sigma\|^2_{H_{N+1}}  & \lesssim \left(  \varepsilon^2 I_{N+2}^2 +  \varepsilon^3 \Lambda^3 \right)  t^{ \delta}.
\end{align*}
\end{proposition} 

Applying \eqref{eq-evolution-1}--\eqref{eq-evolution-2}, \eqref{Gauss-Ricci-hat-k}, we obtain the estimates for $\dtau \Sigma$, $\dtau \eta$, $g-\gamma$, $R_{ij} + 2 g_{ij }$. In summary, we have
\begin{align*}
t^2 \|\dtau \eta \|^2_{H_{N+1}}  & \lesssim \varepsilon^2 I_{N+2}^2 +  \varepsilon^4 \Lambda^4, \\
t^2 \|\dtau \Sigma_{i j}\|^2_{H_{N+1}} & \lesssim \left(  \varepsilon^2 I_{N+2}^2 +  \varepsilon^3 \Lambda^3 \right)  t^{\delta}, \\
\|g_{ij} - \gamma_{i j}\|^2_{H_{N+1}} & \lesssim \varepsilon^2 I_{N+2}^2 +  \varepsilon^3 \Lambda^3, \\
 t^{2} \| R_{ij} + 2 g_{ij } \|^2_{H_{N}} & \lesssim \left(  \varepsilon^2 I_{N+2}^2 +  \varepsilon^3 \Lambda^3 \right)  t^{\delta}.
\end{align*}

Finally, we can improve the regularity for $\eta$.
\begin{proposition}\label{prop-improve-w}
Under the bootstrap assumptions \eqref{wBT-E-H}--\eqref{wBT},  we obtain
\als{
t \| \eta\|_{H_{N+2}} & \lesssim \varepsilon I_{N+2}+ \varepsilon^2 \Lambda^2, \quad N \geq 2.
%\|g_{ij} - \gamma_{i j} \|_{H_{N+2}} & \lesssim  \varepsilon I_{N+2}.
}
\end{proposition}

\begin{proof}
We follow the proof leading to Proposition \ref{pro-improve-eta}.
With the help of the bootstrap assumptions \eqref{wBT-E-H}--\eqref{wBT}, \eqref{E:energy-inequ-teta} changes into, in the case $m=0$,
\als{
t  \| \tilde \eta_{N+2} \|_{L^2}  \lesssim {} & \varepsilon I_{N+2} + \varepsilon^2 \Lambda^2 + \int_{t_0}^t \varepsilon \Lambda t^{-2+\delta} \cdot t \| \tilde \eta_{N+2} \|_{L^2} \, \di t,
}
where $\tilde \eta_{N+2}$ is defined by \eqref{def-ti-eta}. The Gr\"{o}nwall's inequality yields that
\[
t \|\tilde \eta_{I_{N+2}} \|_{L^2} \lesssim \varepsilon I_{N+2}+ \varepsilon^2 \Lambda^2.
\]
In view of \eqref{def-ti-eta} and the fact $$\|\nabla_{I_{N+2}} \eta\|_{L^2} \lesssim  \|\tilde \eta_{N+2}\|_{L^2} +  \varepsilon^2 I^2_{N+2} t^{-1} + \varepsilon^2 \Lambda^2 t^{-2+2\delta},$$ we conclude the estimate.
\begin{comment}
The proof for $\|g_{ij} - \gamma_{i j} \|_{H_{N+2}}$ follows in an analogous way. Noting that
\als{
& \dtau  \nabla_{I_N} \Delta (g_{ij} - \gamma_{i j} ) = -2  \nabla_{I_N} \Delta \eta g_{ij} - 2  \nabla_{I_N} \Delta \Sigma_{ij},% \\
% ={}&  -2  \nabla_{I_N} \Delta \eta g_{ij} - 2  \nabla_{I_N} \dtau^2 \Sigma_{ij} - 4  \nabla_{I_N} \dtau \Sigma_{ij} -6\nabla_{I_N} \nabla_i \nabla_j \eta +2 \nabla_{I_N} \Delta \eta g_{i j} -2 \nabla_{I_N} \dtau \R_{ij} (\phi),
}
making use the wave equation for $\Sigma$, to arrive at

Then
\als{
& \dtau \left( \nabla_{I_N} \Delta (g_{ij} - \gamma_{i j} ) + 2 \nabla_{I_N} (\nabla_i \phi \nabla_j 
\phi) + 2  \nabla_{I_N} \dtau \Sigma_{ij} + 4  \nabla_{I_N}  \Sigma_{ij} \right)  \\
 ={}& -6  \nabla_{I_N}  \nabla_i \nabla_j \eta +  \nabla_{I_N} \left( (\Sigma, \eta)* \dtau \Sigma_{ij} \right) +  \nabla_{I_N}\left( (\Sigma, \eta)* (\Sigma_{ij}, \eta) \right) +  \nabla_{I_N} \left( D \phi * D \phi \right) + l.o.t.,
}
and

\als{
 & t \p_t  \| \nabla_{I_N} \Delta (g_{ij} - \gamma_{i j} ) + 2 \nabla_{I_N} (\nabla_i \phi \nabla_j \phi) + 2  \nabla_{I_N} \dtau \Sigma_{ij} + 4  \nabla_{I_N}  \Sigma_{ij} \|_{L^2} \\
  \lesssim {} & \|\eta\|_{H_{N+2}} +  \varepsilon^2 I_{N+2}^2 t^{-2+2\sqrt{\sigma}}  +  \varepsilon^2 I_{N+2}^2  t^{-2+2\delta}.
}
Then
\[
\| \nabla_{I_N} \Delta (g_{ij} - \gamma_{i j} ) \|_{L^2} \lesssim  \varepsilon I_{N+2}.
\]
By the density result and Proposition \ref{prop-elliptic-Delta}, we prove the claim.
\end{comment}
\end{proof}

\subsubsection{Proof of Theorem \ref{ME-w}} The estimates in propositions \ref{prop-ee-wave}--\ref{prop-improve-w} improve the bootstrap assumptions \eqref{wBT-E-H}--\eqref{wBT}. %\ref{prop-energy-estimate-Bianchi-w} with \ref{prop-energy-estimate-2nd-w}-\ref{prop-improve-w}, 
Thus we conclude there is a constant $C(I_{N+2})$ depending on $I_{N+2}$ such that
\als
{
t ( \|\Ew \|_{H_{N}} & +  \|\Hw \|_{H_{N}}   +  \|\eta \|_{H_{N+2}} ) (t)  + t^{1- \delta/2} \|\Sigma \|_{H_{N+1}} (t) \nnb \\ 
&+   t^{1-\sqrt{\sigma}} ( \|\dtau \phi \|_{H_{N+1}} + \| \phi \|_{H_{N+2}}) (t)   \leq \varepsilon C(I_{N+2}).
}
By virtue of the local existence theorem \ref{thm-local-existence}, we complete the proof of the main theorem \ref{ME-w}.

\appendix
\section{Local existence}\label{sec-local}

The Einstein scalar field system \eqref{eq-Einstein-source}--\eqref{def-energy-Mom-kg} over $(\mathcal{M}, \, \breve g)$ in the geodesic polar gauge ($\breve g = -dt^2 + \tilde g$) consists of the evolution equations
\begin{subequations}
\begin{align}
\p_t \tilde g_{ij} & = -2 \tilde k_{i j}, \label{pt-g} \\
\p_t \tilde k_{ij}& = \tilde R_{ij} - 2 \tilde k_i^p \tilde k_{jp} + \text{tr}_{\tilde g} \tilde k \tilde k_{ij} - \Tr_{ij} (\phi), \label{pt- k} \\
0 &=\Box_{\breve g} \phi - m^2 \phi, \label{kg}
\end{align}
\end{subequations}
and the constraint equations
\begin{subequations}
\begin{align}
\tilde R -|\tilde k|^2 + (\tr_{\tilde g} \tilde k)^2  &= \Tr_{tt}(\phi) - \tr_{\breve g} \breve{T}(\phi), \label{Constrain-Guass} \\
\tilde \nabla^i \tilde k_{ij} - \tilde \nabla_j \tr_{\tilde g} \tilde k &= - \breve{T}_{tj}(\phi), \label{Constrain-Codazzi}
\end{align}
\end{subequations}
where 
\begin{equation}\label{eq-ricci-kg-phi}
\Tr_{\alpha \beta}(\phi):= \breve{T}_{\alpha \beta} -  \frac{\tr_{\breve g} \breve{T}}{2} \breve g_{\alpha \beta}= \breve D_\alpha \phi \breve D_\beta \phi + \frac{m^2}{2}\phi^2 \breve g_{\alpha \beta}.
\end{equation}

Taking $\p_t$ derivative on \eqref{pt- k}, one obtains a wave type equation for $\tilde k$.
Let $h = \tr_{\tilde g} \tilde k$ be a new variable. We end up with the following reduced system \cite{Fournodavlos-Luk-Kasner},
\al{reduced-sys}{
\p_t \tilde g_{ij} & = -2 \tilde k_{i j}, \nnb \\
\p_t h & = |\tilde k|^2 + \Tr_{t t} (\phi), \nnb \\
-\p_t^2 \tilde k_{ij} + \Delta_{\tilde g} \tilde k_{ij}  &= \tilde \nabla_i \tilde \nabla_j h + \p_t \Tr_{ij}(\phi) - \nabla_i (\p_t \phi \nabla_i \phi) - \nabla_j (\p_t \phi \nabla_i \phi) \nnb \\
&+ \tilde k* \tilde k*\tilde k + \p_t \tilde k * \tilde k + (\Tr_{mn}(\phi), \Tr_{t t}(\phi)) * \tilde k , \nnb \\
%-\p_t^2 \phi + \Delta_{\tilde g} \phi - m^2 \phi  &= \tilde k * \p_t \phi.
\Box_{\breve g} \phi - m^2 \phi &=0.
}

Following a computation analogous to \cite{Fournodavlos-Luk-Kasner}, and combining with the divergence-free property of the energy-momentum tensor, $\breve D^\mu \breve{T}_{\mu \nu}= 0$, one can show that the reduced system \eqref{E:reduced-sys} and  the original Einstein coupled with matter field system \eqref{pt-g}--\eqref{Constrain-Codazzi} are equivalent, if the data are those induced from \eqref{pt-g}--\eqref{Constrain-Codazzi}.
\begin{lemma}\label{equi-reduce-EKG}
Suppose $(\tilde g, h, \tilde k, \phi)$ is a solution of the reduced system \eqref{E:reduced-sys} whose initial data $(\tilde g, h, \tilde k, \p_t \tilde k, \phi, \p_t \phi)|_{t=t_0} =(\tilde g_0, h_0, \tilde k_0, \tilde k_1, \phi_0, \phi_1)$ satisfies the original constraint equations \eqref{Constrain-Guass}--\eqref{Constrain-Codazzi}, and $h_0 = \tr_{\tilde g_0} \tilde k_0$, %and $(g, h, \tilde k, \phi)$ satisfies $h-h_0 \in H_{1}(M, g), \tr_g \tilde k - h_0 \in H_{1}(M, g), \phi - \phi_0 \in H_{s+2}(M, g), \p_t\phi - \phi_1 \in H_{s+1}(M, g), g-g_0 \in H_{s+3}(M, g), \tilde k-\tilde k_0 \in H_{s+2}(M, g), \p_t \tilde k-\tilde k_1 \in H_{s+1}(M, g), $
then $\breve g = -dt^2 + \tilde g$ is a solution to the original Einstein scalar field system \eqref{pt-g}--\eqref{Constrain-Codazzi}.
\end{lemma}

Then one can use the reduced system \eqref{E:reduced-sys} to prove the local existence theorem, which had been carried out in \cite{Fournodavlos-Luk-Kasner} in the space $(\tilde g, h, \tilde k, \p_t \tilde k) \in H_{N+2} \times H_{N+2} \times H_{N+1} \times H_{N}$, $N\geq 2$. However, for our practice, we need to get rid of the top ($N+2$) order derivative of $\tilde g$, since it has energy growth in the long time scheme of the massive case. Motivated by the ideas in \cite{Christodoulou-K-93} and \cite{Fournodavlos-Luk-Kasner}, we establish the following local existence theorem in which the $H_{N+2}$ bounded condition for $\tilde g$ will be replaced by the $H_N$ bound on the Ricci curvature.

\begin{theorem}\label{thm-local-existence}
Let $(M, g_0)$ be a smooth complete Riemannian manifold with positive injective radius and $M$ is diffeomorphic to $\mathbb{R}^3$. Let $(g_0, k_0, \phi_0, \phi_1)$ be the data  on $\{t_0\} \times M$, $t_0 >0$, for the rescaled Einstein scalar field system  and satisfies the following conditions: Fix an integer $N\geq 2$,
\begin{itemize}
\item [1] $Ric_{0}$, the Ricci curvature of $g_{0}$, satisfies\footnote{This condition implies that the Sobolev embeddings are valid on $(M,\, g_0)$.} $Ric_0 \geq \alpha g_0$, for some $\alpha \in \mathbb{R}$. In addition, $\|g_0 -\gamma\|_{H_{N+1}(M, g_0)}$ is bounded\footnote{This condition is imposed to assure the density theorem $H_{0,N+2} (C^\infty(M), g_0) =  H_{N+2} (C^\infty(M), g_0)$ holds, in view of Proposition \ref{pro-density}. It can be replaced by $Ric_{0} \in C^{N} (M)$. }.
%\item [2]  $H_{0,N+2} (C^\infty(M), g_0) =  H_{N+2} (C^\infty(M), g_0)$\footnote{This condition will be satisfied if %$g$ and $\gamma$ are equivalent, and $\|g_0 -\gamma\|_{H_{N+1}(M, g_0)}$ is bounded, by Proposition \ref{pro-density}.}.
\item [2] $k_{0ij}$ is a symmetric $(0,2)$-tensor decomposing into the trace free and trace parts: $k_{0ij} = \Sigma_{0ij} + \frac{\tr_{g_0} k_0}{3} g_0$, with $\tr_{g_0} k_0 = g_0^{i j} k_{0 ij}$.
It verifies that $\Sigma_{0} \in H_{N+1}(M, g_0)$ and $ \tr_{g_0} k_0 +3  \in  H_{N+2}(M, g_0)$%\footnote{Here the number $3$ can be replaced by any other number after rescaling the metric.  If it is replaced by $ \tr_{g_0} k_0 + 3 a \in  H_{N+2}(M, g_0)$ for some $a \in \mathbb{R}$ here, then the requirement for $R_0$  in condition 4 should be replaced by $R_0 + 6a^2 \in H_{N}(M, g_0)$, due to the constraint equations. We add this constant to cover the case (for instance, the hyperbolic metric on the hyperboloid $\mathbb{H}^3$) when $R_0\notin H_{N}(M, g_0)$, but $R_0 + c \in H_{N}(M, g_0)$ for some $c \geq 0, \, c \in \mathbb{R}$. },
\item [3] The Ricci tensor satisfies $Ric_{0} + 2g_0 \in H_{N}(M, g_0)$.
\item [4] $(\phi_0, \phi_1)$  $\in H_{N+2}(M, g_0) \times H_{N+1}(M, g_0)$.
\end{itemize}
Then there is a unique, local-in-time development $(\mathcal{M}, \, \breve g)$ with 
\[\mathcal{M} = [t_{0}, t_{\ast}] \times M, \quad \breve g = -dt^2 + t^2 g(t),\] 
and $t =t_0$ corresponding to the initial slice $(M, \, t_0^2 \cdot g_0)$. Moreover, denoting $\Sigma_{ij}$, the trace free part of $k_{ij}$, and $R_{ij}$, the Ricci curvature of $g(t)$ respectively,
we have
\als{
g_{ij}(t) -g_{0 ij} & \in C^1([t_0, t_\ast], H_{N+1}(M, g_0)), \\
\Sigma_{ij} (t) & \in C^1([t_0, t_\ast], H_{N+1}(M, g_0)), \\
 \tr_g k  (t) +3 & \in C^1([t_0, t_\ast], H_{N+2}(M, g_0)),   \\
R_{ij} (t) + 2 g_{ij} (t)& \in C^1([t_0, t_\ast], H_{N}(M, g_0)).
}
\end{theorem}
The proof is standard \cite{Christodoulou-K-93} (see \cite{Fournodavlos-Smulevici-local} as well), except that we need additionally use the idea in \cite{Fournodavlos-Luk-Kasner} (or referring to Section \ref{sec-improve-regularity}) to improve the regularity of $\tr_g k$. %Note that, to derive the higher order estimates, we should use Proposition \ref{prop-elliptic-Delta} to handle the regularity problem, since the metric has regularity that is one order lower than the usual one.%Alternatively, one can follow \cite[Section 7]{Klanierman-Nicolo-review}, using the Bianchi equations as in the long time scheme, to formulate an analogously local theorem. 

\section{The density theorem}\label{sec-density}

We recall the following density theorem from \cite{Hebey}.
\begin{proposition}\label{prop-density}
Let $(M, g)$ be a smooth, complete Riemannian manifold, then the following statement holds:  

\begin{itemize}
\item  For any $p \geq 1$, $H_{0,1}^p(M) = H^p_1(M)$.
\item Assume that $(M, g)$ has positive injective radius, and $|\nabla^j R_{m n}|$, $j = 0, \cdots, K-2$, is bounded, where $K \geq 2$ is an integer. Then for any $p \geq 1$, $H_{0,K}^p(M) = H^p_K(M)$.
\end{itemize}
\end{proposition}

We will make use of the above results to establish a density theorem for our purpose.

\begin{proposition}\label{pro-density}
Let $(M, g)$ be a smooth, complete Riemannian manifold with positive injective radius and the Ricci curvature is bounded from below. We fix an integer $N \geq 2$.

 Suppose $M$ is diffeomorphic to $\mathbb{R}^3$, and let  $\gamma_{ij}$ be the hyperbolic metric on $M$. %Suppose $g_{ij}$ and $\gamma_{ij}$ are equivalent as bilinear forms: there is some constant $C>0$, such that $C^{-1} \gamma_{ij} \leq g_{ij} \leq C \gamma_{i j}$. Additionally, 
Assume that \[g_{ij} - \gamma_{i j} \in H_{N+1} (M, g), \quad N \geq 2. \] Then
\al{density-top}
{H_{0,k}(M)  = H_{k}(M),  & & k \leq N + 1, \nnb \\
H_{0, l} (C^\infty(M))  = H_{l} (C^\infty(M)), & & l \leq N + 2.
}
%where $H_{N+2} (C^\infty(M))$ denotes the completion of the space of smooth functions with bounded $H_{N+2}$ norm with respect to the $H_{N+2}$ norm, and $H_{0, N+2} (C^\infty(M))$ denotes the closure of the space of smooth functions with compact support (in $M$) in $H_{N+2}(C^\infty(M))$.

\end{proposition}
\begin{proof}
 Since %$g_{ij}$ and $\gamma_{i j}$ are equivalent as bilinear forms, and 
$g_{ij} - \gamma_{i j} \in H_{N+1}$, $N \geq 2$, noting that $\nabla_\gamma g = \nabla \gamma * g$, hence there is, 
\als{
\sum_{1\leq k \leq N+1}\| \nabla^{I_k} \gamma_{i j}\|_{L^2(M, g)} + \sum_{1 \leq k\leq N+1} \| \nabla^{I_k}_\gamma g_{i j}\|_{L^2(M, \gamma)} & \leq C.
}
%and hence $\sum_{1 \leq k\leq N+1} \| \nabla_{I_k} [\gamma] g_{i j}\|_{L^2}  \lesssim \varepsilon I_{N+2}$, 
As a result, applying the Sobolev inequalities, we can prove that for any function $\psi$ and tensor field $\Psi$ on $M$, the two norms are equivalent
\als{
\| \psi\|_{H_{l} (M, g)} & \sim \| \psi\|_{H_{l} (M, \gamma)}, \quad  l \leq N + 2, \\
\| \Psi\|_{H_{k} (M, g)} & \sim \| \Psi\|_{H_{k} (M, \gamma)}, \quad k \leq N + 1, 
}
  and therefore,
\als{
 H_{l} (C^\infty(M), g) &=  H_{l} (C^\infty(M), \gamma), && l \leq N + 2,\\
  H_{0,l} (C^\infty(M), g) &=  H_{0, l} (C^\infty(M), \gamma), &&  l \leq N + 2.
 }
By the density theorem on $(M, \gamma)$ (referring to Proposition \ref{prop-density}), \[H_{0,l} (C^\infty(M), \gamma) =  H_{l} (C^\infty(M), \gamma),\] we conclude the second claim in \eqref{E:density-top}. The first claim in \eqref{E:density-top} follows in the same way.
\end{proof}

\section{Some identities}

\subsection{Commuting identity}\label{sec-comm}
Let $\Gamma_{ij}^a$ be the connection coefficient of $\nabla$. Then the Lie derivative $\lie_{\dtau} \Gamma_{ij}^a$ is a tensor field
\begin{equation}\label{dtau-Gamma}
\begin{split}
\lie_{\dtau} \Gamma^a_{ij}& =\frac{1}{2}g^{ab} \left( \nabla_i \lie_{\dtau} g_{jb} + \nabla_j \lie_{\dtau} g_{ib} -\nabla_b \lie_{\dtau} g_{ij} \right).
\end{split}
\end{equation}
%We remind ourselves that $\lie_{\dtau} g_{ij} = -2 \eta g_{ij} - 2 \Sigma_{ij}.$

A commuting identity between $\nabla$ and $\lie_{\dtau}$ is given below:
\begin{lemma}\label{lem-commu-lie}
Let $\Psi$ be an arbitrary $(0, k)$-tensor field on $(M, g)$. The following commuting formula holds:
\begin{equation}\label{commuting-lie-nabla}
\lie_{\dtau} \nabla_j \Psi_{a_1 \cdots a_k} = \nabla_j \lie_{\dtau} \Psi_{a_1 \cdots a_k} - \sum_{i=1}^k \lie_{\dtau} \Gamma^p_{j a_i} \Psi_{a_1 \cdots p \cdots a_k}.
\end{equation}
\end{lemma}
This lemma can be proved by straightforward calculations.
An application of Lemma \ref{lem-commu-lie} to $\nabla_{I_l}  \psi$, taking \eqref{eq-evolution-1} into account, gives the following lemma.
\begin{lemma}\label{lemma-commuting-application}%[Commuting lemma for scalar field]
Let $l \geq 1$,  then for any scalar field $\psi$
\begin{align}
  \lie_{\dtau}  \nabla_{I_l} \psi &=  \nabla_{I_l} \left( \dtau  \psi \right) +  [ \dtau, \nabla_{I_l}](\psi), \label{id-commuting-nabla-N-T-l-simplify}
\end{align}
where $[ \dtau, \nabla_{I_1}](\psi)=0,$ and 
\begin{align}
  [\dtau, \nabla_{I_l}](\psi) &= \sum_{a+2 +b  = l}  \nabla_{I_a}  \left( \nabla \Sigma_{i j}, \, \nabla \eta \right) * \nabla_{I_{b}} \nabla \psi, \quad l \geq 2. \label{def-commuting-KN-l}
\end{align}
\end{lemma}

We also present a commuting identity between $\nabla$ and $\Delta$, which can be proved by induction. 
\begin{lemma}\label{lemma-commuting-nabla-laplacian}
For any scalar field $\psi$ and  $l \geq 1$,
\begin{align}
\Delta \nabla_{I_l} \psi &=  \nabla_{I_l} \Delta \psi  +  \sum_{a+b= l-1} \nabla_{I_a} R_{imjn} * \nabla_{I_b} \nabla \psi. \label{def-R-l-commute-nabla-laplacian}
\end{align}
In general, for any $(0, n)$-tensor $\Psi_{J_n},$
\begin{align}
 \Delta \nabla_{I_l} \Psi_{J_n} &= \nabla_{I_l} \Delta \Psi_{J_n} +  \sum_{a+b = l} \nabla_{I_a} R_{imjn} * \nabla_{I_b} \Psi_{J_n}. \label{Def-R-Psi-ij-commute-nabla-laplacian}
\end{align}
In particular,
\al{comm-2}{
&\nabla_a \Delta \Psi_{I_n} = \Delta \nabla_a \Psi_{I_n} - R_a^p \nabla_p \Psi_{I_n} \nnb\\
& \quad + \sum_{k=1}^n 2 R_{a p i_k}{}^{i_q} \nabla^p \Psi_{i_1 \cdots i_q \cdots i_n} + \sum_{k=1}^n  \nabla^p R_{a p i_k}{}^{i_q} \Psi_{i_1 \cdots i_q \cdots i_n}.
}
\end{lemma}

Using these commuting identities, we can prove Proposition \ref{prop-elliptic-Delta}.
\begin{proof}[Proof of Proposition \ref{prop-elliptic-Delta}]
When $k=1$, it holds automatically.

When $k=2$, we take a $(0,1)$-tensor $\Psi$ for instance.
 \begin{align*}
& \int_{M} |\nabla^2 \Psi |^2\,  \di \mu_g = -  \int_{M} \nabla^j \Psi^p \Delta \nabla_j  \Psi_p \, \di \mu_g \\
& \stackrel{\eqref{E:comm-2}}{=} -  \int_{M} \nabla^j \Psi^p \left( \nabla_j \Delta \Psi_{p } + R_{j}^i \nabla_i \Psi_p - 2 R_{j a p}{}^{b} \nabla^a \Psi_{b } -  \nabla^a R_{j a p}{}^{b} \Psi_{b} \right) \di \mu_g \\
& =  \int_{M} \left( |\Delta \Psi|^2 -  R_{j}^i \nabla^j \Psi  \nabla_i \Psi +  R_{j a p}{}^{b} \nabla^a \Psi_{b } \nabla^j \Psi^p - R_{j a p}{}^{b}  \Psi_{b}  \nabla^a \nabla^j \Psi^p \right) \di \mu_g \\
& \lesssim  \int_{M} |\Delta \Psi|^2 \, \di \mu_g + \|R_{imjn}\|_{L^\infty} \|\nabla \Psi\|^2_{L^2} + \| R_{imjn} \|_{L^\infty} \| \Psi\|_{L_2} \|\nabla^2 \Psi\|_{L^2},
 \end{align*}
 where we used integration by parts in the last identity.
 That is,
  \begin{align*}
\|\nabla^2 \Psi\|^2_{L^2} 
%& \lesssim   \|\Delta \Psi\|^2_{L^2} + ( \|R_{imjn}\|_{L^\infty} + \|\nabla R_{imjn} \|_{L^2} ) \| \Psi\|^2_{H_1} \\
%&\quad  +   \|\nabla R_{imjn} \|_{L^2} \| \nabla^2 \Psi\|_{L^2} \| \Psi\|_{H_1} \\
& \leq   \|\Delta \Psi\|^2_{L^2} +  C \|R_{imjn}\|_{L^\infty} \|\Psi\|^2_{H_1} \\
&\quad  +C(  a^{-1} \| R_{imjn} \|^2_{L^\infty}  \| \Psi\|^2_{L^2} + a \| \nabla^2 \Psi\|^2_{L^2} ).
 \end{align*}
 We take the constant $a$ such that $a C <\frac{1}{2}$ to derive
   \be\label{Bochner-2}
\|\nabla^2 \Psi\|^2_{L^2} \lesssim  \|\Delta \Psi\|^2_{L^2} +  ( \|R_{imjn}\|_{L^\infty} + \| R_{imjn} \|^2_{L^\infty} ) \| \Psi\|^2_{H_1}.
 \ee

 In general, for $k \geq 3$, we derive the identity
\begin{align*}
 & \quad \int_{M} \nabla_{I_k} \Psi \nabla^{I_k} \Psi \, \di \mu_g  = -  \int_{M} \nabla_{I_{k-1}} \Psi \Delta \nabla^{I_{k-1}} \Psi \, \di \mu_g \\
& \stackrel{\eqref{Def-R-Psi-ij-commute-nabla-laplacian}}{=} -  \int_{M} \nabla_{I_{k-1}} \Psi \left(  \nabla^{I_{k-1}}  \Delta \Psi +  \sum_{a+b= k-1} \nabla_{I_a} R_{imjn} * \nabla_{I_b}  \Psi  \right) \di \mu_g \\
& =   \int_{M} \Delta \nabla_{I_{k-2}} \Psi  \nabla^{I_{k-2}}  \Delta \Psi \, \di \mu_g + \int_M  \sum_{a+b= k-1} \nabla_{I_a} R_{imjn} * \nabla_{I_b}  \Psi * \nabla_{I_{k-1}} \Psi \, \di \mu_g \\
& \stackrel{\eqref{Def-R-Psi-ij-commute-nabla-laplacian}}{=}   \int_{M} |\nabla^{I_{k-2}}  \Delta \Psi |^2 \, \di \mu_g + \int_M \sum_{a+b= k-2} \nabla_{I_a} R_{imjn} * \nabla_{I_b}  \Psi * \nabla^{I_{k-2}}  \Delta \Psi \, \di \mu_g \\
& \quad + \int_M \sum_{a+b= k-1} \nabla_{I_a} R_{imjn} * \nabla_{I_b}  \Psi * \nabla_{I_{k-1}} \Psi \, \di \mu_g.
\end{align*}
Applying integration by parts again  to $$\int_M \sum_{a+b= k-2} \nabla_{I_a} R_{imjn} * \nabla_{I_b}  \Psi * \nabla^{I_{k-2}}  \Delta \Psi \, \di \mu_g,$$ we have
\begin{align*}
\|\nabla_{I_k} \Psi\|^2_{L^2}  
& =  \| \nabla_{I_{k-2}} \Delta \Psi \|^2_{L^2}  + \int_M \sum_{a+b= k-1} \nabla_{I_a} R_{imjn} * \nabla_{I_b}  \Psi * \nabla^{I_{k-3}}  \Delta \Psi \, \di \mu_g \\
& \quad + \int_M \sum_{a+b= k-1} \nabla_{I_a} R_{imjn} * \nabla_{I_b}  \Psi * \nabla_{I_{k-1}} \Psi  \, \di \mu_g.
\end{align*}
For terms like $$\int_M \nabla_{I_{k-1}} R_{imjn} *  \Psi * \nabla_{I_{k-1}} \Psi \, \di \mu_g,$$ we apply integration by parts to obtain
\begin{align*}
 & \int_M \nabla_{I_{k-1}} R_{imjn} *  \Psi * \nabla_{I_{k-1}} \Psi \, \di \mu_g \\
 = & \int_M \nabla_{I_{k-2}} R_{imjn} * \nabla \Psi * \nabla_{I_{k-1}} \Psi + \nabla_{I_{k-2}} R_{imjn} * \Psi * \nabla_{I_{k}} \Psi \, \di \mu_g  \\
  = & \int_M \nabla_{I_{k-3}} \nabla R_{imjn} * \nabla \Psi * \nabla_{I_{k-1}} \Psi + \nabla_{I_{k-3}} \nabla R_{imjn} * \Psi * \nabla_{I_{k}} \Psi \, \di \mu_g.
\end{align*}
Thus, for $k \geq 3$, we have 
\begin{align*}
 \|\nabla_{I_k} \Psi \|^2_{L^2} & \lesssim  \| \nabla_{I_{k-2}} \Delta \Psi \|^2_{L^2}  +  \|R_{imjn} \ast  \nabla_{I_{k-1}} \Psi \ast  \nabla_{I_{k-1}} \Psi \|_{L^1}  \\
 & + \sum_{a+b=k-3}  \|\nabla_{I_a} \nabla R_{imjn} \ast \nabla_{I_b} \nabla \Psi \ast  \nabla_{I_{k-1}} \Psi \|_{L^1}    \\
 & + \| \nabla_{I_{k-3}} \nabla R_{imjn} * \Psi * \nabla_{I_{k}} \Psi \|_{L^1}.
\end{align*}
%Note that 
%\als{
%&\sum_{a+b=k-2}  \|\nabla_{I_a} \nabla R_{imjn} \cdot \nabla_{I_b} \Psi \|_{L^2}  \\
%\leq & \sum_{a+b=k-3}  \|\nabla_{I_a} \nabla R_{imjn} \cdot \nabla_{I_b} \nabla \Psi \|_{L^2}  +  \|\nabla_{I_{k-1}} R_{imjn} \cdot \Psi \|_{L^2} 
%}
Then
\begin{align*}
  \|\nabla_{I_k} \Psi \|^2_{L^2} & \lesssim  \| \nabla_{I_{k-2}} \Delta \Psi \|^2_{L^2} +  \|R_{imjn}\|_{L^\infty}  \| \Psi \|^2_{H_{k-1}} \\
 & + \sum_{a+b=k-3, \, a < k-3}  \|\nabla_{I_a} \nabla R_{imjn} \|_{L^4} \| \nabla_{I_b} \nabla \Psi \|_{L^4} \|\nabla_{I_{k-1}} \Psi\|_{L^2}   \\
 & + \| \nabla_{I_{k-2}} R_{imjn} \|_{L^2} \| \nabla \Psi \|_{L^4} \|\nabla_{I_{k-1}} \Psi\|_{L^4}  \\
 & +  \|\nabla_{I_{k-2}} R_{imjn} \|_{L^2} \|\Psi \|_{L^\infty}  \| \nabla_{I_{k}} \Psi \|_{L^2}.
\end{align*}
Using the Sobolev inequalities (noting that $k \geq 3$),
\begin{align*}
 \|\nabla_{I_k} \Psi \|^2_{L^2} & \lesssim  \| \nabla_{I_{k-2}} \Delta \Psi \|^2_{L^2} +  \|R_{imjn}\|_{L^\infty}  \| \Psi \|^2_{H_{k-1}} \\
 & +   \|  \nabla R_{imjn} \|_{H_{k-3}} \| \nabla \Psi \|_{H_{k-2}}  \| \Psi \|_{H_{k-1}} \\
 & +  a^{-1} \|\nabla_{I_{k-2}}  R_{imjn} \|^2_{L^2} \|\Psi \|^2_{H_2} + a \| \nabla_{I_{k}} \Psi \|^2_{L^2}.
\end{align*}
Choosing $a$ to be small so that $a \| \nabla_{I_{k}} \Psi \|^2_{L^2}$ can be absorbed by the left hand side of the above inequality, for $k\geq 3$, we have, 
\als{
  \|\nabla_{I_k} \Psi \|^2_{L^2} & \lesssim  \| \nabla_{I_{k-2}} \Delta \Psi \|^2_{L^2}  + \|R_{imjn}\|_{L^\infty}   \| \Psi \|^2_{H_{k-1}} \nnb \\
 & \quad + \left( \| \nabla R_{imjn} \|_{H_{k-3}} + \| \nabla R_{imjn} \|^2_{H_{k-3}} \right)  \| \Psi \|^2_{H_{k-1}}.
}
By induction,
\al{Bochner-k}
 {\|\nabla_{I_k} \Psi \|^2_{L^2}  \lesssim {}& \| \nabla^{\mathring{k}} \Delta^{ [\frac{k}{2}] } \Psi \|_{L^2} \nnb\\
 &+ C \left( \|R_{imjn}\|_{L^\infty}, \, \| \nabla R_{imjn} \|_{H_{N-1}} \right)  \| \Psi \|^2_{H_{k-1}},
}
 for all $0 \leq k\leq N+2$, where $C \left( \|R_{imjn}\|_{L^\infty}, \, \| \nabla R_{imjn} \|_{H_{N-1}} \right)$ is a constant depending on $\|R_{imjn}\|_{L^\infty}$ and $\| \nabla R_{imjn} \|_{H_{N-1}}$.

\end{proof}

\subsection{An identity for the Bianchi equations}\label{sec-id-Bianchi}
We will prove an identity for the Bianchi equations (Lemma \ref{lemma-div-curl}), which is crucial in high order energy estimates for the $1+3$ Bianchi equations.

\begin{proof}[Proof of Lemma \ref{lemma-div-curl}]
Appealing to the commuting identity \eqref{E:comm-2}, we can prove this lemma by induction.

For the proof, it suffices to keep track of the principle part of $R_{imjn}$ (without derivatives). By Gauss equation \eqref{E:Gauss-Riem-hat-k}, we know that the principle part of $R_{imjn}$ is given by $-\left(g_{ij} g_{mn} -g_{in} g_{mj} \right).$ %In what follows, we shall introduce the notation $\simeq$, which means equalling in the principle part. For example, $R_{imjn}\simeq -\left(g_{ij} g_{mn} -g_{in} g_{mj} \right).$

For $k=0$, we have
\begin{align}\label{pre-case}
& \quad \curl \nabla^{\mathring{k}} \Delta^{ [\frac{k}{2}] } H \cdot \nabla^{\mathring{k}} \Delta^{ [\frac{k}{2}] } E  -  \curl \nabla^{\mathring{k}} \Delta^{ [\frac{k}{2}] } E  \cdot \nabla^{\mathring{k}} \Delta^{ [\frac{k}{2}] } H \nnb \\
& =  g^{ab} \epsilon_i{}^{\! pq}   \nabla_q  \nabla^{\mathring{k}}_a \Delta^{ [\frac{k}{2}] } H_{pj} \cdot  \nabla^{\mathring{k}}_b \Delta^{ [\frac{k}{2}] } E^{ij}  - g^{ab} \epsilon_i{}^{\! pq} \nabla_q   \nabla^{\mathring{k}}_a \Delta^{ [\frac{k}{2}] } E_{pj} \cdot  \nabla^{\mathring{k}}_b \Delta^{ [\frac{k}{2}] }  H^{ij}  \nnb \\
&=  \nabla_q \left(g^{ab} \epsilon_i{}^{\! pq}  \nabla^{\mathring{k}}_a \Delta^{ [\frac{k}{2}] }  H_{pj} \cdot  \nabla^{\mathring{k}}_b \Delta^{ [\frac{k}{2}] } E^{ij} \right)   - g^{ab} \epsilon^{ipq} \nabla_q  \nabla^{\mathring{k}}_a \Delta^{ [\frac{k}{2}] } E_{ij} \cdot   \nabla^{\mathring{k}}_b \Delta^{ [\frac{k}{2}] } H_p^{j}  \nnb\\
& \quad \quad -  g^{ab} \epsilon^{ipq} \nabla_q  \nabla^{\mathring{k}}_a \Delta^{ [\frac{k}{2}] } E_{pj} \cdot  \nabla^{\mathring{k}}_b \Delta^{ [\frac{k}{2}] }H_i^{j}  \nnb  \\
&=   \nabla_q \left( g^{ab} \epsilon_i{}^{\! pq}  \nabla^{\mathring{k}}_a \Delta^{ [\frac{k}{2}] } H_{pj} \cdot  \nabla^{\mathring{k}}_b \Delta^{ [\frac{k}{2}] } E^{ij} \right).
\end{align}
To prove the case $k=1$, note that
\als{
 & \epsilon_i{}^{\! pq} \nabla_{i_1} \nabla_q  H_{pj} \cdot \nabla^{i_1} E^{ij}  - \epsilon_i{}^{\! pq} \nabla_{i_1} \nabla_q  E_{pj} \cdot \nabla^{i_1} H^{ij} \\
 ={}&  \epsilon_i{}^{\! pq}  \nabla_q \nabla_{i_1} H_{pj} \cdot \nabla^{i_1} E^{ij}  - \epsilon_i{}^{\! pq} \nabla_q   \nabla_{i_1} E_{pj} \cdot \nabla^{i_1} H^{ij} \\
 & +   \epsilon_i{}^{\! pq}  R_{i_1 q p}{}^{m} H_{m j} \cdot \nabla^{i_1} E^{ij} +   \epsilon_i{}^{\! pq}  R_{i_1 q j}{}^{n} H_{p n} \cdot \nabla^{i_1} E^{ij} \\
  &-   \epsilon_i{}^{\! pq}  R_{i_1 q p}{}^{m} E_{m j} \cdot \nabla^{i_1} H^{ij} -  \epsilon_i{}^{\! pq}  R_{i_1 q j}{}^{n} E_{p n} \cdot \nabla^{i_1} H^{ij} \\
= {} &  \epsilon_i{}^{\! pq}  \nabla_q \nabla_{i_1} H_{pj} \cdot \nabla^{i_1} E^{ij}  - \epsilon_i{}^{\! pq} \nabla_q   \nabla_{i_1} E_{pj} \cdot \nabla^{i_1} H^{ij} \\
  & - \epsilon_{i i_1}{}^{\! a}  H_{a j}  \nabla^{i_1} E^{ij} + \epsilon_{i i_1}{}^{\! a} E_{a j} \nabla^{i_1} H^{ij} \\
  & + O_{a b m n} * H * \nabla E + O_{a b m n} * E * \nabla H \\
= {} &   \nabla_q \left(  \epsilon_i{}^{\! pq}  \nabla_a  H_{pj} \cdot  \nabla^a  E^{ij} \right) + \nabla^{q} \left( \epsilon_{i q}{}^{\! a} E_{a j}  H^{ij} \right) \\
& + O_{a b m n} * H * \nabla E + O_{a b m n} * E * \nabla H,
}
where we used \eqref{pre-case}  in the last equality.

In general, suppose the conclusion holds  for $k \in \mathbb{Z}_+$,
\al{indu-k-1}{
& \epsilon_i{}^{\! pq}  \Delta^{ k } \nabla_q   H_{pj} \cdot  \Delta^{k } E^{ij}  -  \epsilon_i{}^{\! pq} \Delta^{ k } \nabla_q   E_{pj} \cdot   \Delta^{k }  H^{ij}  
}
and
\al{indu-k-2}{
&\quad  \epsilon_i{}^{\! pq}     \nabla_a \Delta^{k } \nabla_q  H_{pj} \cdot  \nabla^a \Delta^{ k} E^{ij}  -   \epsilon_i{}^{\! pq}   \nabla_a \Delta^{ k } \nabla_q  E_{pj} \cdot  \nabla^a \Delta^{ k}  H^{ij},
}
with arbitrary symmetric tensors $E, \, H$ on $M$. Then we will prove it holds with $k$ replaced by $k+1$. 
 
{\bf Step I.} By the commuting identity \eqref{E:comm-2}, 
\als{
& \epsilon_i{}^{\! pq}  \Delta^{ k+1 } \nabla_q   H_{pj} \cdot  \Delta^{k+1 } E^{ij}   -  \epsilon_i{}^{\! pq} \Delta^{ k+1 } \nabla_q   E_{pj} \cdot   \Delta^{k+1 }  H^{ij}  \\
={}& \epsilon_i{}^{\! pq} \Delta^{ k } \nabla_q  \Delta H_{pj} \cdot  \Delta^{k+1 } E^{ij}  -  \epsilon_i{}^{\! pq} \Delta^{ k} \nabla_q  \Delta E_{pj} \cdot   \Delta^{k+1 }  H^{ij} \\
&- \epsilon_i{}^{\! pq} \Delta^k  (2R_{q a p}{}^{b} \nabla^a H_{bj} + 2R_{qaj}{}^{b} \nabla^a H_{p b} - R_q^a \nabla_a H_{p j}) \cdot  \Delta^{k+1 } E^{ij} \\
&+ \epsilon_i{}^{\! pq} \Delta^k  (2R_{q a p}{}^{b} \nabla^a E_{bj} + 2R_{qaj}{}^{b} \nabla^a E_{p b} - R_q^a \nabla_a E_{p j}) \cdot  \Delta^{k+1 } H^{ij} \\
&+ \Delta^k (\nabla R_{a b m n} * H) * \Delta^{k+1} E + \Delta^k (\nabla R_{a b m n} * E) * \Delta^{k+1} H \\
={}& \epsilon_i{}^{\! pq}   \Delta^{ k } \nabla_q ( \Delta H_{pj} ) \cdot  \Delta^{k  } ( \Delta E^{ij} )  -  \epsilon_i{}^{\! pq} \Delta^{ k} \nabla_q  (\Delta E_{pj}) \cdot   \Delta^{k} ( \Delta H^{ij}) \\
& - 4 \epsilon_i{}^{\! pq} \Delta^k \nabla_p H_{q j} \cdot  \Delta^{k+1 } E^{ij} + 4 \epsilon_i{}^{\! pq} \Delta^k \nabla_p E_{q j} \cdot  \Delta^{k+1 } H^{ij} \\
&+ \Delta^k \nabla ( R_{a b m n} * H) * \Delta^{k+1} E + \Delta^k \nabla ( R_{a b m n} * E) * \Delta^{k+1} H.
}
By induction \eqref{E:indu-k-1}, the first line on the right hand side of the last equality $$\epsilon_i{}^{\! pq} \Delta^{ k } \nabla_q ( \Delta H_{pj} ) \cdot  \Delta^{k  } ( \Delta E^{ij} )  -  \epsilon_i{}^{\! pq} \Delta^{ k} \nabla_q  (\Delta E_{pj}) \cdot   \Delta^{k} ( \Delta H^{ij})$$ can be expressed in the expected form, and the second line
\als{
& - 4 \epsilon_i{}^{\! pq} \Delta^k \nabla_p H_{q j} \cdot  \Delta^{k+1 } E^{ij} + 4 \epsilon_i{}^{\! pq} \Delta^k \nabla_p E_{q j} \cdot  \Delta^{k+1 } H^{ij} \\
={}& - \nabla_m (4 \epsilon_i{}^{\! pq} \Delta^k \nabla_p H_{q j} \cdot  \nabla^m \Delta^{k} E^{ij}) + \nabla_m (4 \epsilon_i{}^{\! pq} \Delta^k \nabla_p E_{q j} \cdot  \nabla^m \Delta^{k} H^{ij})  \\
&+ 4 \epsilon_i{}^{\! pq}  \nabla_m  \Delta^k \nabla_p H_{q j} \cdot  \nabla^m \Delta^{k} E^{ij} - 4 \epsilon_i{}^{\! pq}  \nabla_m  \Delta^k \nabla_p E_{q j} \cdot  \nabla^m \Delta^{k} H^{ij},
}
where the last line is as well in the expected form by induction \eqref{E:indu-k-2}.

{\bf Step II.}
\als{
&  \epsilon_i{}^{\! pq} \nabla_a \Delta^{k+1 } \nabla_q  H_{pj} \cdot  \nabla^a \Delta^{ k+1 } E^{ij} -   \epsilon_i{}^{\! pq}   \nabla_a \Delta^{ k +1 } \nabla_q  E_{pj} \cdot  \nabla^a \Delta^{ k+1}  H^{ij} \\
={}& \epsilon_i{}^{\! pq}  \nabla_a \Delta^{ k } \nabla_q  \Delta H_{pj} \cdot \nabla^a \Delta^{k+1 } E^{ij}  -  \epsilon_i{}^{\! pq} \nabla_a \Delta^{ k} \nabla_q  \Delta E_{pj} \cdot \nabla^a  \Delta^{k+1 }  H^{ij} \\
&- \epsilon_i{}^{\! pq} \nabla_a \Delta^k  (2R_{q a p}{}^{b} \nabla^a H_{bj} + 2R_{qaj}{}^{b} \nabla^a H_{p b} - R_q^a \nabla_a H_{p j}) \cdot \nabla^a \Delta^{k+1 } E^{ij} \\
&+ \epsilon_i{}^{\! pq} \nabla_a \Delta^k  (2R_{q a p}{}^{b} \nabla^a E_{bj} + 2R_{qaj}{}^{b} \nabla^a E_{p b} - R_q^a \nabla_a E_{p j}) \cdot \nabla^a \Delta^{k+1 } H^{ij} \\
&+ \nabla \Delta^k (\nabla R_{a b m n} * H) * \nabla \Delta^{k+1} E + \nabla \Delta^k (\nabla R_{a b m n} * E) * \nabla \Delta^{k+1} H \\
={}& \epsilon_i{}^{\! pq}   \nabla_a \Delta^{ k } \nabla_q ( \Delta H_{pj} ) \cdot \nabla^a \Delta^{k  } ( \Delta E^{ij} )  -  \epsilon_i{}^{\! pq} \nabla_a \Delta^{ k} \nabla_q  (\Delta E_{pj}) \cdot  \nabla^a \Delta^{k} ( \Delta H^{ij}) \\
& - 4 \epsilon_i{}^{\! pq} \nabla_a \Delta^k \nabla_p H_{q j} \cdot \nabla^a \Delta^{k+1 } E^{ij} + 4 \epsilon_i{}^{\! pq} \nabla_a \Delta^k \nabla_p E_{q j} \cdot \nabla^a \Delta^{k+1 } H^{ij} \\
&+ \nabla \Delta^k \nabla ( O_{a b m n} * H) * \nabla \Delta^{k+1} E + \nabla \Delta^k \nabla ( O_{a b m n} * E) * \nabla \Delta^{k+1} H,
}
where by induction \eqref{E:indu-k-2}, the first line $$\epsilon_i{}^{\! pq}   \nabla_a \Delta^{ k } \nabla_q ( \Delta H_{pj} ) \cdot \nabla^a \Delta^{k  } ( \Delta E^{ij} )  -  \epsilon_i{}^{\! pq} \nabla_a \Delta^{ k} \nabla_q  (\Delta E_{pj}) \cdot  \nabla^a \Delta^{k} ( \Delta H^{ij})$$ can be recast into the expected form. Now we are left with the second line:
\als{
& - 4 \epsilon_i{}^{\! pq} \nabla_a \Delta^k \nabla_p H_{q j} \cdot \nabla^a \Delta^{k+1 } E^{ij} + 4 \epsilon_i{}^{\! pq} \nabla_a \Delta^k \nabla_p E_{q j} \cdot \nabla^a \Delta^{k+1 } H^{ij} \\
={}& - \nabla^a (4 \epsilon_i{}^{\! pq} \nabla_a \Delta^k \nabla_p H_{q j} \cdot  \Delta^{k+1} E^{ij}) + \nabla^a (4 \epsilon_i{}^{\! pq} \nabla_a \Delta^k \nabla_p E_{q j} \cdot  \Delta^{k+1} H^{ij})  \\
&+ 4 \epsilon_i{}^{\! pq}   \Delta^{k+1} \nabla_p H_{q j} \cdot  \Delta^{k+1} E^{ij} - 4 \epsilon_i{}^{\! pq}   \Delta^{k+1} \nabla_p E_{q j} \cdot \Delta^{k+1} H^{ij}
}
which has been already confirmed in the previous step.
\end{proof}

\end{document}